\documentclass[a4paper,11pt]{article}
\usepackage[margin=1in]{geometry}
\usepackage{complexity} 
\usepackage{enumerate,paralist}
\usepackage{amsmath}
\usepackage{amssymb,amsthm,complexity}
\usepackage{authblk}
\usepackage{todonotes}

\newcommand{\sv}[1]{}

\usepackage{vmargin}
\setmarginsrb{1in}{1in}{1in}{1in}{0mm}{0mm}{0mm}{7mm}

\usepackage{mathtools}
\usepackage{thmtools}
\usepackage{thm-restate}
\usepackage{hyperref}
\usepackage{cleveref}
\usepackage{complexity}
\usepackage{authblk}
\theoremstyle{plain}
\usepackage{todonotes}
\usepackage{microtype}
\usepackage{multirow}
\usepackage{amsmath,amssymb}
\usepackage{comment}
\usepackage{enumerate}
\usepackage{xspace}
\usepackage{listings}
\usepackage{color}
\usepackage{cite}
\usepackage{graphicx}
\usepackage{subcaption}
\RequirePackage{fancyhdr}
 \usepackage{xcolor}
\usepackage{boxedminipage}
\usepackage{algorithm2e}
\usepackage{tabularx}
\usepackage{booktabs}
\usepackage{pgfplots}

\newcommand{\defproblem}[3]{
  \vspace{3mm}
\noindent\fbox{
  \begin{minipage}{.95\textwidth}
  \begin{tabular*}{\textwidth}{@{\extracolsep{\fill}}lr} #1  \\ \end{tabular*}
  {\bf{Input:}} #2  \\
  {\bf{Question:}} #3
  \end{minipage}
  }
  \vspace{2mm}
  }

\newtheorem{theorem}{\bf Theorem}
\newtheorem{definition}{\bf Definition}
\newtheorem{proposition}{\bf Proposition}

\newtheorem{corollary}{\bf Corollary}

\newtheorem{observation}{\bf Observation}
\newtheorem{reduction rule}{\bf Reduction Rule}
\newtheorem{branching rule}{\bf Branching Rule}
\newtheorem{lemma}{\bf Lemma}

\newtheorem{remark}{\bf Remark}

\newtheorem{innercustomthm}{Procedure}
\newenvironment{reduction procedure}[1]
{\renewcommand\theinnercustomthm{#1}\innercustomthm}
  {\endinnercustomthm}

\newcommand{\MonOneThreeSAT}{{\sc Monotone $1$-in-$3$ SAT}}

\newcommand{\vwsp}{{\sc Valued WSP}}
\newcommand{\vapep}{{\sc Valued APEP}}

\newcommand{\ev}[1]{\langle #1 \rangle}

\newcommand{\userProfile}{user profile}
\newcommand{\uP}{\operatorname{usr}}
\newcommand{\uPs}[2]{\ensuremath{\uP_{#1}(#2)}}
\newcommand{\uProfSet}{\mathcal{P}}

\newcommand{\maxdiff}{\rm maxdiff}

\newcommand{\hA}{\hat{A}}

\newcommand{\sE}{{\sf E}}

\newcommand{\sU}{{\sf U}}

\newcommand{\OPT}{\mbox{\rm OPT}}

\newcommand{\bodU}{{\sf BoD_{\sU}}}
\newcommand{\bodE}{{\sf BoD_{\sE}}}

\newcommand{\sodU}{{\sf SoD_{\sU}}}
\newcommand{\sodE}{{\sf SoD_{\sE}}}

\newcommand{\bN}{{\mathbb N}}

\newcommand{\cI}{{\mathcal I}}
\newcommand{\cO}{{\mathcal O}}

\newcommand{\BoDU}[1]{( #1 ,\leftrightarrow, \forall)}
\newcommand{\BoDE}[1]{( #1 ,\leftrightarrow, \exists)}
\newcommand{\SoDE}[1]{( #1 , \updownarrow, \exists)}
\newcommand{\SoDU}[1]{( #1 , \updownarrow, \forall)}

\newcommand{\card}[1]{( ,\leq, )}

\title{Valued Authorization Policy Existence Problem: Theory and Experiments}

\date{}

\author[1]{Jason Crampton}
\author[1]{Eduard Eiben}
\author[1]{Gregory Gutin}
\author[2]{Daniel Karapetyan}
\author[1]{Diptapriyo Majumdar}
\affil[1]{Royal Holloway, University of London, Egham, United Kingdom\\
    \texttt{\{jason.crampton|eduard.eiben|g.gutin|diptapriyo.majumdar\}@rhul.ac.uk}}
\affil[2]{University of Nottingham, United Kingdom\\
  \texttt{daniel.karapetyan@nottingham.ac.uk}}

\begin{document}

\maketitle

\begin{abstract}
Recent work has shown that many problems of satisfiability and resiliency in workflows may be viewed as special cases of the authorization policy existence problem (APEP), which returns an authorization policy if one exists and ``No'' otherwise. However, in many practical settings it would be more useful to obtain a ``least bad'' policy than just a ``No'', where ``least bad'' is characterized by some numerical value indicating the extent to which the policy violates the base authorization relation and constraints.
Accordingly, we introduce the Valued APEP, which returns an authorization policy of minimum weight, where the (non-negative) weight is determined by the constraints violated by the returned solution.

We then establish a number of results concerning the parameterized complexity of Valued APEP.
We prove that the problem is fixed-parameter tractable (FPT) if the set of constraints satisfies two restrictions, but is intractable if only one of these restrictions holds.
(Most constraints known to be of practical use satisfy both restrictions.)

We also introduce a new type of resiliency for workflow satisfiability problem, show how it can be addressed using Valued APEP and use this to build a set of benchmark instances for Valued APEP.
Following a set of computational experiments with two mixed integer programming (MIP) formulations, we demonstrate that the Valued APEP formulation based on the user profile concept has FPT-like running time and usually significantly outperforms a naive formulation.
\end{abstract}




\maketitle

\section{Introduction}\label{sec:intro}

Access control is a fundamental aspect of the security of any multi-user computing system.
Access control requirements are typically enforced by specifying an authorization policy and implementing a system to enforce the policy.
Such a policy identifies which interactions between users and resources are to be allowed (and denied) by the access control system.

Over the years, authorization policies have become more complex, not least because of the introduction of constraints -- often motivated by business requirements such as ``Chinese walls'' -- which further refine an authorization policy.
A separation-of-duty constraint (also known as the ``two-man rule'' or ``four-eyes policy'') may, for example, require that no single user is authorized for some particularly sensitive group of resources.
Such a constraint is typically used to prevent misuse of the system by a single user.

The use of authorization policies and constraints, by design, limits which users may access resources.
Nevertheless, the ability to perform one's duties will usually require access to particular resources, and overly prescriptive policies and constraints may mean that some resources are unavailable to users that need access to them.
In other words, there may be some conflict between authorization policies and operational demands: a policy that is insufficiently restrictive may suit operational requirements but lead to security violations; conversely, too restrictive a policy may compromise an organization's ability to meet its business objectives.

Recent work on workflow satisfiability and access control resiliency recognized the importance of being able to determine whether or not security policies prevent an organization from achieving its objectives~\cite{CrGuWa15,CrGuYe13,LiWaTr09,MaMoMo14,WaLi10}.
Berg\'e \emph{et al.} introduced the {\sc Authorization Policy Existence Problem} (APEP)~\cite{BergeCGW18}, which generalizes many of the existing satisfiability and resiliency problems in access control.
Informally, the APEP seeks to find an authorization policy, subject to restrictions on individual authorizations (defined by a base authorization relation) and restrictions on collective authorizations (defined by a set of authorization constraints).

APEP may be viewed as a decision or search problem.
An algorithm to solve either version of the problem returns ``no'' if no authorization policy exists, given the base authorization relation and the constraints that form part of the input to the instance.
Such a response is not particularly useful in practice: from an operational perspective, an administrator would presumably find it more useful if an algorithm to solve APEP returned some policy, even if that policy could lead to security violations, provided the risk of deploying that policy could be quantified in some way.

Hence, in this paper, we introduce a generalization of APEP, which we call {\sc Valued APEP}, where every policy is associated with a non-negative weight.
A solution to {\sc Valued APEP} is a policy of minimum weight; a policy of zero weight satisfies the base authorization relation and all the constraints.

We establish the complexity of {\sc Valued APEP} for certain types of constraints, using multi-variate complexity analysis.
We prove that {\sc APEP} (and hence {\sc Valued APEP}) is fixed-parameter intractable, even if all the constraints are user-independent, a class of constraints for which the {\sc Workflow Satisfiability Problem} (WSP) -- a special case of APEP -- is fixed-parameter tractable.
However, we subsequently show that {\sc Valued APEP} is fixed-parameter tractable when all weighted constraints are user-independent and the set of constraints is $t$-weight-bounded ($t$-wbounded).
Informally, the identities of the users are irrelevant to the solution and there exists a solution (policy) containing no more than $t$ authorizations.
We show that sets of user-independent constraints that contain only particular kinds of widely used constraints are $t$-wbounded. 
Berg\'e \emph{et al.}~\cite{BergeCGW18} introduced and used a notion of a bounded constraint. Bounded and wbounded constraints have some similarities, but wbounded constraints are more refined and allow for more precise complexity analysis.
In particular, the notion of a bounded constraint cannot be used for {\sc Valued APEP} and we are able to derive improved complexity results for APEP using wbounded constraints.

A significant innovation of the paper is to introduce the notion of a user profile for a weighted constraint.
Counting user profiles provides a powerful means of analyzing the complexity of ({\sc Valued}) APEP, somewhat analogous to the use of patterns in the analysis of workflow satisfiability problems.
This enables us to \begin{inparaenum}[(i)]\item derive the complexity of {\sc Valued APEP} when all constraints are $t$-wbounded and user-independent, \item establish the complexity of {\sc Valued} APEP for the most common types of user-independent constraints, and \item improve on existing results for the complexity of APEP obtained by Berg\'e \emph{et al.}~\cite{BergeCGW18}\end{inparaenum}.
We also prove that our result for the complexity of APEP with $t$-wbounded, user-independent constraints cannot be improved, unless a well-known and widely accepted hypothesis in parameterized complexity theory is false.
Finally, we show that certain sub-classes of {\sc Valued APEP} can be reduced to the {\sc Valued Workflow Satisfiability Problem} ({\sc Valued WSP})~\cite{CramptonGK15} with user-independent constraints, thereby establishing that these sub-classes are fixed-parameter tractable.

For the first time in the APEP literature, we conduct computational experimentsbased on an application of \vapep{} in WSP\@.
Specifically, we introduce a concept of $\tau$-resiliency in WSP which seeks a solution that is resilient to deleting up to $\tau$ arbitrary users.
We build a set of \vapep{} benchmark instances that address $\tau$-resiliency in WSP and use it in computational experiments.
To solve \vapep{}, we use two mixed integer programming formulations.
One is a `naive' formulation of the problem whereas the other one exploits the user profile concept.
We demonstrate that the formulation based on the user profile concept has FPT-like solution times and usually outperforms the naive formulation by a large margin.
We also analyse and discuss the properties of $\tau$-resiliency.

In the next section, we summarize relevant background material.
We introduce the {\sc Valued APEP} in Section~\ref{sec:vapep} and define weighted user-independent constraints.
We also show that {\sc Valued WSP} is a special case of {\sc Valued APEP} and describe particular types of weighted user-independent constraints for APEP.
In Section \ref{sec:bound}, introduce the notion of a $t$-wbounded constraint and establish the complexity of {\sc Valued APEP} when all constraints are $t$-wbounded.
We prove the problem is intractable for arbitrary sets of $t$-wbounded constraints or user-independent constraints, but fixed-parameter tractable for $t$-wbounded, user-independent constraints.
In the following two sections, we establish the complexity of other sub-classes of {\sc Valued APEP}.
In Section~\ref{sec:resilient-wsp} we explain how APEP can be used to address questions of resiliency in workflows.
In the following two sections we introduce two MIP formulations for \vapep{} and test these formulations using the resiliency questions introduced in Section~\ref{sec:resilient-wsp}.
In Section~\ref{sec:related-work} we discuss how our contributions improve and extend related work.
We conclude the paper with a summary of our contributions and some ideas for future work in Section~\ref{sec:conclusions}.

\section{Background}\label{sec:background}

APEP is defined in the context of a set of \emph{users} $U$, a set of \emph{resources} $R$,  a \emph{base authorization relation} $\hA \subseteq U \times R$, and a set of \emph{constraints} $C$.
Informally, APEP asks whether it is possible to find an authorization relation $A$ that satisfies all the constraints and is a subset of $\hA$.

For an arbitrary authorization relation $A \subseteq U \times R$ and an arbitrary resource $r \in R$, we write $A(r)$ to denote the set of users authorized for resource $r$ by $A$; more formally, $A(r) = \{u \in U \mid (u, r) \in A\}$.
For a subset $T \subseteq R$, we define the set of users authorized for some resource in $T$ to be $A(T) = \bigcup_{r \in T} A(r)$.
For a user $u$, $A(u) = \{r \in R \mid (u,r) \in A\}$; and for $V \subseteq U$, $A(V)=\{r\in R \mid (u, r) \in A, u\in V\}$.

An authorization relation $A \subseteq U \times R$ is
\begin{itemize}
	\item {\em authorized} (with respect to $\hA$) if $A \subseteq \hA$,
	\item {\em complete} if for all $r \in R$, $A(r) \neq \emptyset$,
	\item {\em eligible} (with respect to $C$) if it satisfies all $c \in C$,
	\item {\em valid} (with respect to $\hA$ and $C$) if $A$ is authorized, complete, and eligible.
\end{itemize}
An instance of APEP is {\em satisfiable} if it admits a valid authorization relation $A$.

\subsection{APEP constraints}\label{sec:apep-constraints}

In general, there are no restrictions on the constraints that can appear in an APEP instance, although the use of arbitrary constraints has a significant impact on the computational complexity of APEP (see Section~\ref{sec:complexity-wsp-apep}).
Accordingly, Berg\'e {\em et al.}~\cite{BergeCGW18} defined several standard types of constraints for APEP, summarized in Table~\ref{tbl:apep-constraints}, generalizing existing constraints in the access control literature.

 \begin{table}[h]\caption{Standard APEP constraints: $r,r' \in R$, $t \in \mathbb{N}$}\label{tbl:apep-constraints}
 \begin{tabular}{|l|l|l|l|}
 \hline
  \bf Description & \bf Notation & \bf Satisfaction criterion & \bf Constraint family \\
  \hline
  Universal binding-of-duty & $(r,r',\leftrightarrow,\forall)$ & $A(r) = A(r')$ & $\bodU$ \\
  Universal separation-of-duty & $(r,r',\updownarrow,\forall)$ & $A(r) \cap A(r') = \emptyset$ & $\sodU$ \\
  Existential binding-of-duty & $(r,r',\leftrightarrow,\exists)$ & $A(r) \cap A(r') \ne \emptyset$ & $\bodE$ \\
  Existential separation-of-duty & $(r,r',\updownarrow,\exists)$ & $A(r) \ne A(r')$ & $\sodE$ \\
Cardinality-Upper-Bound & $(r,\leq,t)$ & $|A(r)| \leq t$ & $\sf Card_{UB}$ \\
Cardinality-Lower-Bound & $(r,\geq,t)$ & $|A(r)| \geq t$ & $\sf Card_{LB}$ \\
  \hline
 \end{tabular}
  \end{table}

\subsection{WSP as a special case of APEP}\label{sec:wsp-special-case-of-apep}

Consider an instance of APEP which contains any of the constraints defined in Section~\ref{sec:apep-constraints}, and includes the set of cardinality constraints $\{(r,\leq,1) \mid r \in R\}$.
Any solution $A$ to such an APEP instance requires $|A(r)| = 1$ for all $r \in R$ (since completeness requires $|A(r)| > 0$).
Thus $A$ may be regarded as a function $A : R \rightarrow U$.
Since $|A(r)| = 1$, there is no distinction between existential and universal constraints (whether they are separation-of-duty or binding-of-duty): specifically, $A$ satisfies the constraint $(r,r',\circ,\exists)$ iff $A$ satisfies $(r,r',\circ,\forall)$ (for $\circ \in \{\updownarrow,\leftrightarrow\}$).

In other words, an APEP instance of this form is equivalent to an instance of WSP~\cite{WaLi10,CrGuYe13}, with separation-of-duty, binding-of-duty and cardinality constraints: resources correspond to workflow steps, the base authorization relation to the authorization policy, and an APEP solution to a plan.
%
%
Accordingly, strong connections exist between APEP and WSP, not least because certain instances of APEP can be reduced to WSP~\cite{BergeCGW18}. In WSP, the set of resources is the set of {\em steps}, denoted by $S.$

\subsection{Complexity of WSP and APEP}\label{sec:complexity-wsp-apep}
In the context of WSP, the authorization policy (the base authorization relation in APEP) specifies which users are authorized for which steps in the workflow.
A solution to WSP is a plan $\pi$ that assigns a single user to each step in the workflow.
In general, WSP is {\sf NP}-complete \cite{WaLi10}. 

Let $k=|S|$ and $n=|U|.$ 
Then there are $n^k$ plans, and the validity of each plan can be established in polynomial time (in the size of the input). 
Thus WSP can be solved in polynomial time if $k$ is constant. 
It is easy to establish that APEP is harder than WSP in general.

\begin{proposition}
 {\sc APEP} is {\sf NP}-complete even when there is a single resource.
\end{proposition}

We provide a polynomial time reduciton from {\MonOneThreeSAT}~\cite{CSP78} problem.
We formally state the problem definition.

\defproblem{{\MonOneThreeSAT}}{A $3$-CNF formula $\phi$ such that no literal is a negated variable.}{Does $\phi$ have a satisfying assignment that assigns {\sc True} to only one literal from every clause?}

The proof uses a simple reduction from {\sc Monotone 1-in-3 SAT}~\cite{CSP78} to an instance of APEP in which there is a single resource $r$: the set of variables corresponds to the set of users; $(x,r) \in A$ corresponds to assigning the value {\sc True} to variable $x$; and every clause corresponds to a constraint comprising three ``users'', which is satisfied provided exactly one user is assigned to the resource.

%

Wang and Li~\cite{WaLi10} introduced parameterization\footnote{We provide a brief introduction to parameterized complexity in Section \ref{sec:parameterized-complexity}.} of WSP by parameter $k$. 
This parameterization is natural because for many practical instances of WSP, $k=|S|\ll n=|U|$ and $k$ is relatively small.  
Wang and Li proved that WSP is intractable, even from the parameterized point of view. 
However, Wang and Li proved that WSP becomes computationally tractable from the parameterized point of view (i.e., fixed-parameter tractable) when the constraints are restricted to some generalizations of binary separation-of-duty (SoD) and binding-of-duty (BoD) constraints. 

Similarly, for APEP, we denote $k=|R|$ and $n=|U|.$ In the rest of the paper, we assume that $k$ is relatively small and thus consider it as the parameter. 
While the assumption that $k$ is small is not necessarily correct in some applications, our approach is useful where $k$ is indeed small, for example in special cases such as WSP.
Also, there are situations where strict controls are placed on the utilization of and access to (some small subset of system) resources by users.

\subsection{User-independent constraints}
Wang and Li's result has been extended to the much larger family of user-independent constraints, which includes the aforementioned SoD and BoD constraints and most other constraints that arise in practice~\cite{CoCrGaGuJo14,KarapetyanPGG19}. 
Informally, a constraint is called user-independent if its satisfaction does not depend on the identities of the users assigned to steps. 
(For example, it is sufficient to assign steps in a separation of duty constraint to different users in order to satisfy the constraint.)

The concept of a user-independent constraint for WSP can be extended formally to the APEP setting in the following way~\cite{BergeCGW18}. 
Let $\sigma: U \rightarrow U$ be a permutation on the user set.
Then, given an authorization relation $A \subseteq U \times R$, we write $\sigma(A) = \{(\sigma(u), r)|(u, r) \in A\}$.
A constraint $c$ is said to be {\em user-independent} if for every authorization relation $A$ that satisfies $c$ and every permutation $\sigma: U \rightarrow U$, $\sigma(A)$ also satisfies $c$.
It is not hard to see that the sets of constraints defined in Section~\ref{sec:apep-constraints} are user-independent~\cite{BergeCGW18}, since their satisfaction is independent of the specific users that belong to $A(r)$ and $A(r')$.

Berg\'e {\em et al.} established a number of FPT results for APEP (restricted to $t$-bounded, user-independent constraints).
We introduce a definition of user-independence and $t$-boundedness for weighted constraints in Sections~\ref{sec:vapep} and~\ref{sec:bound}, respectively, and show that we can improve on existing complexity results.

\subsection{Parameterized complexity}\label{sec:parameterized-complexity}
An instance of a parameterized problem $\Pi$ is a pair $(I,\kappa)$ where $I$ is the {\em main part} and $\kappa$ is the {\em parameter}; the latter is usually a non-negative integer.  
A parameterized problem is {\em fixed-parameter tractable} ({\sf FPT}) if there exists a computable function $f$ such that any instance $(I,\kappa)$ can be solved in time $O(f(\kappa)|{I}|^c)$, where $|I|$ denotes the size of~$I$ and $c$ is an absolute constant. 
The class of all fixed-parameter tractable decision problems is called {{\sf FPT}} and algorithms which run in the time specified above are called {{\sf FPT}} algorithms. 
As in other literature on {{\sf FPT}} algorithms, we will often omit the polynomial factor in $\cO(f(\kappa)|{I}|^c)$ and write $\cO^*(f(\kappa))$ instead.

Consider two parameterized problems $\Pi$ and $\Pi'$. 
We say that $\Pi$ has a {\em parameterized reduction} to $\Pi'$ if there are functions $g$ and $h$ from $\mathbb{N}$ to $\mathbb{N}$
and a function $(I,\kappa)\mapsto (I',\kappa')$ from $\Pi$ to $\Pi'$ such that 
\begin{itemize}
 \item  there is an algorithm of running time  $h(\kappa) \cdot (|I|+\kappa)^{O(1)}$ which for input
 $(I,\kappa)$ outputs $(I',\kappa')$, where $\kappa'\le g(\kappa)$;     and 
 \item $(I,\kappa)$ is a yes-instance of $\Pi$ if and only if $(I',\kappa')$ is a yes-instance of $\Pi'$.
\end{itemize}

While {\sf FPT} is a parameterized complexity analog of {\sf P} in classical complexity theory, there are many parameterized hardness classes, forming a nested sequence of which {\sf FPT} is the first member: {\sf FPT}$\subseteq$ {\sf W}[1]$\subseteq$ {\sf W}[2 $]\subseteq \dots$. The {\em Exponential Time Hypothesis} (ETH) is a well-known and plausible conjecture that there is no algorithm solving 3-CNF Satisfiability in time $2^{o(n)}$, where $n$ is the number of variables \cite{IP}.
It is well known that if the ETH holds then ${\sf FPT} \ne {\sf W}[1]$. 
Hence, {\sf W}[1] is generally viewed as a parameterized intractability class, which is an analog of {\sf NP} in classical complexity.  

A well-known example of a  {\sf W}[1]-complete problem is the {\sc Clique} problem parameterized by $\kappa$: given a graph $G$ and a natural number $\kappa$, decide whether $G$ has a complete subgraph on $\kappa$ vertices. 
A well-known example of a {\sf W}[2]-complete problem is the {\sc Dominating Set} problem parameterized by $\kappa$: given a graph $G=(V,E)$ and a natural number $\kappa$, decide whether $G$ has a set $S$ of $\kappa$ vertices such that every vertex in $V\setminus S$ is adjacent to some vertex in $S.$
Thus, every {\sf W}[1]-hard problem $\Pi_1$ is at least as hard as {\sc Clique} (i.e., {\sc Clique} has a parameterized reduction to $\Pi_1$); similarly, every {\sf W}[2]-hard problem $\Pi_2$ is at least as hard as {\sc Dominating Set}.

More information on parameterized algorithms and complexity can be found in recent books~\cite{cygan2015,downey2013}.
%
%
%
 
 \section{Valued APEP}\label{sec:vapep}
 
 As we noted in the introduction, we believe that it is more valuable, in practice, for APEP to return some authorization relation, even if that relation is not valid (in the sense defined in Section~\ref{sec:background}).
 Clearly, the authorization relation that is returned must be the best one, in some appropriate sense.
 Inspired by {\sc Valued WSP}, we introduce {\sc Valued APEP}, where every authorization relation is associated with a ``cost'' (more formally, a \emph{weight}) and the solution to a {\sc Valued APEP} instance is an authorization relation of minimum weight.
 
 \subsection{Problem definition}
 We first introduce the notions of a \emph{weighted constraint} and a \emph{weighted user authorization function}.
 Let $A \subseteq U \times R$ be an authorization relation.
 A weighted constraint $c$ is defined by a function $w_c : 2^{U \times R} \rightarrow \mathbb{N}$ such that $w_c(A) = 0$ if and only if $A$ satisfies the constraint.
 Hence, we we will use interchangeably $c$ and $w_c$ as a
notation for $c.$
 By definition, $w_c(A) > 0$ if the constraint is violated.
 The intuition is that $w_c(A)$ represents the cost incurred by $A$, in terms of constraint violation.
 For example, a weighted constraint $w_c$ such that $w_c(A) = 0$ iff $A(r) \cap A(r') = \emptyset$ and $w_c(A)$ increases monotonically with the size of $A(r) \cap A(r')$ encodes the usual APEP constraint $(r,r',\updownarrow,\forall)$.
 (We describe other weighted constraints in Section~\ref{sec:valued-apep-constraints}.) 
 
 A weighted user authorization function $\omega : U \times 2^R \rightarrow \mathbb{N}$ has the following properties:
 \begin{align}
  &\omega(u,T) = 0 \quad\text{if $u$ is authorized for each resource in $T$} \\
  \label{moncond} &T' \subseteq T \quad \text{implies} \quad \omega(u,T') \leq \omega(u,T).
 \end{align}
 Then $\omega(u,T) > 0$ if $u$ is not authorized for some resource in $T$ and, vacuously, we have $\omega(u,\emptyset) = 0$ for all $u \in U$.
 The weighted user authorization function is used to represent the cost of assigning unauthorized users to resources.

 Then we define the weighted authorization function $\Omega : 2^{U \times R} \rightarrow \mathbb{N}$, weighted constraint function $w_C : 2^{U \times R} \rightarrow \mathbb{N}$, and weight function $w : 2^{U \times R} \rightarrow \mathbb{N}$ as follows:
 \begin{align}
  \label{eq:Omega} \Omega(A) &= \sum_{u \in U} \omega(u,A(u)), \\
  w_C(A) &= \sum_{c \in C} w_c(A), \\
  w(A) &= \Omega(A) + w_C(A).
 \end{align}
 A relation $A$ is \emph{optimal} if $w(A) \leq w(A')$ for all $A' \subseteq U \times R$.
 We now formally define {\sc Valued APEP}. 
  
 \begin{center}
  \fbox{%
  \begin{tabularx}{.95\columnwidth}{X}
   {\sc Valued APEP}\\
   {\bf Input}: A set of resources $R$, a set of users $U$, a set of weighted constraints $C$, a weighted user authorization function $\omega$\\
   {\bf Parameter}: $|R| = k$\\
   {\bf Output}: A complete, optimal authorization relation 
  \end{tabularx}
  }
\end{center}

 \begin{remark}
  A base authorization relation $\hA$ is implicitly defined in a {\sc Valued APEP} instance: specifically, $(u,r) \in \hat{A}$ iff $\omega(u,r) = 0$.
  An instance of {\sc Valued APEP} is defined by a tuple $(R,U,C,\omega)$, where $C$ is a set of weighted constraints; we may, when convenient, refer to $\hat{A}$, as defined by $\omega$. 
 \end{remark}

 \subsection{Valued APEP constraints}\label{sec:valued-apep-constraints}
 We now provide some examples of weighted constraints, extending the examples introduced in Section~\ref{sec:apep-constraints}.
 First, let $f_c : \mathbb{Z} \rightarrow \mathbb{N}$ be a monotonically increasing function (i.e., $f_c(z) \leq f_c(z+1)$ for all $z \in \mathbb{Z}$), where $f_c(z) = 0$ iff $z \leq 0$, and let $\ell_c$ be some constant.
 Define ${\rm maxdiff}(A,r,r')$ to be $\max\{|A(r) \setminus A(r')|,|A(r') \setminus A(r)|\}$.
 Then the equations below demonstrate how an unweighted APEP constraint $c$ may be extended to a weighted constraint $w_c$ using $f_c$.
    \begin{alignat}{2}
     \nonumber & \textbf{Unweighted} \qquad && \textbf{Weighted} \\
     \label{w_c:card} & (r,\leq,t) \qquad && w_c(A) = f_c(|A(r)|-t) \\
      \label{w_c:cardLB} & (r,\geq,t) \qquad && w_c(A) = f_c(t-|A(r)|) \\
     \label{w_c:SoDU} & (r,r',\updownarrow,\forall) \qquad && w_c(A) = f_c(|A(r) \cap A(r')|) \\
     \label{w_c:BoDU} & (r,r',\leftrightarrow,\forall) \qquad && w_c(A) =f_c({\rm maxdiff}(A,r,r')) \\
     \label{w_c:SoDE} & (r,r',\updownarrow,\exists) \qquad &&    w_c(A) = \begin{cases} 0 & \text{if $A(r) \ne A(r')$},\\ \ell_c & \text{otherwise}.
            \end{cases}\\
     \label{w_c:BoDE} & (r,r',\leftrightarrow,\exists) \qquad &&     w_c(A) = \begin{cases} 0 & \text{if $A(r) \cap A(r') \ne \emptyset$},\\ \ell_c & \text{otherwise}.
             \end{cases}
    \end{alignat}
   
   For example, the weighted cardinality constraint~\eqref{w_c:card} evaluates to $0$ if $A$ assigns no more than $t$ users to $r$, and some non-zero value determined by $f_c$ and $|A(r)|$ otherwise.
   The specific choice of function $f_c$ and the constant $\ell_c$ will vary, depending on the particular application and particular constraint that is being encoded.  
   For notational convenience, we may refer to binding-of-duty and separation-of-duty constraints of the form $(r,r',\updownarrow,\forall)$, $(r,r',\leftrightarrow, \forall)$, $(r,r',\updownarrow,\exists)$ and $(r,r',\leftrightarrow,\exists)$.
   However, when doing so, we mean the relevant weighted constraint as defined in equations~\eqref{w_c:SoDU}, \eqref{w_c:BoDU}, \eqref{w_c:SoDE} and~\eqref{w_c:BoDE}, respectively.
 
   Given an authorization relation $A \subseteq U \times R$, we say a weighted constraint $w_c$ is \emph{user-independent} if, for every permutation $\sigma$ of $U$, $w_c(A) = w_c(\sigma(A))$.
   We have already observed that the APEP constraints in Section~\ref{sec:apep-constraints} are user-independent.
   It is easy to see that the weighted constraints defined above for {\sc Valued APEP} are also user-independent.

  In the remaining sections of this paper, we consider the fixed-parameter tractability of {\sc Valued APEP}.
 We will write, for example, $\text{\sc APEP}\langle \bodE \rangle$ to denote the set of instances of APEP in which the set of constraints $C$ contains only $\bodE$ constraints.

 \subsection{Valued APEP and Valued WSP}
   We have already observed that WSP is a special case of APEP for certain choices of APEP constraints.
   Berg\'e {\em et al.} also proved that the complexity of some sub-classes of APEP can be reduced to WSP~\cite[Section~5]{BergeCGW18}.
   
   The inputs to {\sc Valued WSP} include a weighted authorization policy and weighted constraints, and the solution is a plan of minimum weight~\cite{CramptonGK15}.
   Similar arguments to those presented in Section~\ref{sec:wsp-special-case-of-apep} can be used to show that {\sc Valued WSP} is a special case of {\sc Valued APEP}.
   In this paper, we will show that some sub-classes of {\sc Valued APEP} can be reduced to {\sc Valued WSP}, thereby establishing, via the following result~\cite[Theorem~1]{CramptonGK15}, that those sub-classes of {\sc Valued APEP} are FPT.

 \begin{theorem}\label{thm1}\sloppy
  {\sc Valued WSP}, when all weighted constraints are user-independent, can be solved in time \mbox{$\mathcal{O}^*(2^{k\log k})$}, where $k=|S|$. 
 \end{theorem}

\section{$t$-bounded constraints}\label{sec:bound}

In this section we consider instances of \vapep\ having an optimal solution $A^*$ that is small; i.e., $|A^*|\ll |U\times R|$. 
We start by defining a natural restriction on weighted constraints that implies the existence of a small optimal solution for instances containing only constraints satisfying the restriction.
Moreover, checking whether a constraint satisfies the restriction is often easier than checking for the existence of a small optimal solution. 
This restriction roughly says that if the size of an authorization relation is larger than $t$, then there are authorizations that are redundant, in the sense that removing those authorizations does not increase the cost of the authorization relation. 

\begin{definition}[$t$-wbounded]
	A set of weighted constraints $C$ is \emph{$t$-wbounded} if and only if for each complete authorization relation $A$ such that $|A|>t$ there exists a complete authorization relation $A'$ such that $A'\subseteq A$, $|A'|<|A|$, and $w_C(A')\le w_C(A)$.  We say that a weighted constraint \(w_c\) is \(t\)-wbounded if the set \(\{w_c\}\) is \(t\)-wbounded.
\end{definition}
We remark that Berg\'e {\em et al.}~\cite{BergeCGW18} introduced the notion of $f(k,n)$-bounded user-independent constraints for APEP. While they introduced the notion only for the user-independent constraints it can be easily generalized for any APEP constraint as follows. For an authorization relation $A$ and a user $u$ let us denote by $A-u$ the authorization relation obtained from $A$ by removing all the pairs that include the user $u$ (i.e., the relation $A\setminus \{(u,r)\mid r\in R \}$).
\begin{definition}
	Given a set of resources $R$ and a set of users $U$, a constraint $c$ is $f(k,n)$-bounded if for each complete authorization relation $A$ which satisfies $c$, there exists a set $U'$ of size at most $f(k,n)$ such that for each user $u\in (U\setminus U')$, the authorization relation $A-u$ is complete and satisfies $c$.
\end{definition}

One way to generalize $f(k,n)$-bounded constraints to \vapep\ would be to say that a weighted constraint $w_c$ is $f(k,n)$-bounded if for each complete authorization relations $A$ there exists a set $U'$ of at most $f(k,n)$ users such that for every user $u\in U\setminus U'$, the relation $A-u$ is complete and $w_c(A')\le w_c(A)$. 
Given this we can show that our definition covers all constraints covered by Berg\'e~et~al.
\begin{lemma}\label{lem:bounded_implies_wbounded}
	If a weighted constraints $w_c$ is $f(k,n)$-bounded, then $w_c$ is $f(k,n)\cdot k$-wbounded. Moreover, if every $c\in C$ is user-independent and $f(k,n)$-bounded then $C$ is $f(k,n)\cdot 2^k\cdot k$-wbounded.
\end{lemma}

\begin{proof}
	Let us consider a complete relation $A$. If $|A| > f(k,n)\cdot k$, then there are at least $f(k,n)+1$ users authorized by $A$. It follows that there exists a user $u$ such that $A(u)\neq \emptyset$ and the authorization relation $A'=A-u$ is complete and $w_c(A')\le w_c(A)$. But $A'\subseteq A$ and $|A'|<|A|$. Hence $w_c$ is $(f(k,n)\cdot k)$-wbounded. Now, if $|A|> f(k,n)\cdot 2^k\cdot k$, then for some $T\subseteq R$, $T\neq \emptyset$ there are at least $f(k,n)+1$ users $u$ such that $A(u)=T$. Since every $c\in C$ is user-independent and $f(k,n)$-bounded, it is not difficult to see that for a user $u$ with $A(u)=T$ the authorization relation $A' = A-u$ is complete and $w_c(A')\leq w_c(A)$ for all $c\in C$. Therefore $C$ is $f(k,n)\cdot 2^k\cdot k$-wbounded.
\end{proof}

We can now show that if the set of all constraints in an input instance is $t$-wbounded, then the size of some optimal solution is indeed bounded by $t$.

\begin{lemma}\label{lem:existance_of_small_solution}
	Let $\cI = (R, U, C, \omega)$ be an instance of {\vapep} such that $C$ is $t$-wbounded. Then there exists an optimal solution $A^*$ of ${\cI}$ such that $|A^*|\le t$.
\end{lemma}

\begin{proof}
	Let $A$ be an optimal solution of ${\cI}$ that minimizes $|A|$. If $|A|\le t$, then the result follows immediately. For the sake of contradiction, let us assume that $|A|>t$. Since $A$ is a solution, it is complete. Hence by the definition of $t$-wboundedness, it follows that there exists a complete authorization relation $A'\subseteq A$ such that  $A'\subseteq A$, $|A'|<|A|$, and $w_C(A')\le w_C(A)$. 
	Since $A'$ is a complete authorization relation, it follows that $A'$ is also a solution. Because $A'\subseteq A$, it follows that for all $u\in U$ we have $A'(u)\subseteq A(u)$ and by the monotonicity condition on $\omega$ and the definition of the function $\Omega$, we have that $\Omega(A')\le \Omega(A)$. Finally it follows that $w(A')\le w(A)$ and $A'$ is also an optimal solution. This however contradicts the choice of the optimal solution $A$ to be an optimal solution that minimizes~$|A|$. 
\end{proof}

Recall that in WSP the solution is a plan that assigns each step to exactly one user. Hence, we can easily translate an instance of WSP into an APEP instance such that each constraint can be satisfied only if each resource is authorized for exactly one user. Let us call such constraints WSP constraints. 
It follows that if a relation $A\subseteq U\times R$ satisfies a WSP constraint $c$, then $|A|=k$ and there 
are at most $k$ users authorized by $A$. It follows that in an instance of APEP obtained by a straightforward reduction from WSP we have that every constraint is $k$-bounded and the set of all constraints is $k$-wbounded.
Therefore, the W[1]-hardness result for WSP established by Wang and Li~\cite{WaLi10} immediately translates to W[1]-hardness of APEP (and hence also \vapep) parameterized by the number of resources $k$ even when the set of all constraints is $k$-wbounded.

\begin{theorem}
	APEP is W[1]-hard even when restricted to the instances such that $C$ is $k$-wbounded and every constraint of $C$ is $k$-bounded.
\end{theorem}

Given the above hardness result, from now on we will consider only user-independent constraints. 
We first show that the user-independent constraints defined in Section~\ref{sec:valued-apep-constraints} are $t$-wbounded.
 The following lemma significantly improves the existing bounds for {\vapep}$\ev{\bodU, \bodE, \sodE, \sodU, {\sf Card_{UB}}}$ proved in \cite{CramptonEGKM21}.

\begin{lemma}
\label{lemma:quadratic-w-boundedness-bod-sod}
Let $\cI = (R, U, C, \omega)$ be an input instance to {\vapep}$\ev{\bodU, \bodE, \sodE, \sodU, {\sf Card_{UB}}, {\sf Card_{LB}}}$.
Then, $C$ is $3\tau k{{k}\choose{2}}$-wbounded, where 
$\tau = \max_{(r,\ge,t)\in C}t$.
Moreover, for any complete authorization relation $A$, there exists a complete authorization relation $A^* \subseteq A$ such that $A^*$ has at most $3\tau{{k}\choose{2}}$ users and $w_C(A^*) \leq w_C(A)$.
\end{lemma}

\begin{proof}
Let $A$ be a complete authorization relation.
If $A$ has at most $3\tau{{k}\choose{2}}$ users, then {$|A|\le 3\tau k{{k}\choose{2}}$} and we are done.
Suppose that $A$ has more than $3\tau{{k}\choose{2}}$ users.
We define an equivalence relation $\equiv$ in $R$ as follows: $r \equiv r'$ if and only if $A(r) = A(r')$.
This equivalence relation yields a partition of $R$, $R_1 \uplus \ldots \uplus R_p$.
Note that for any $r \in R_i$, we have that $A(R_i) = A(r)$.
For any $i, j \in [p]$ with $i \neq j$, we will consider $A(R_i) \setminus A(R_j), A(R_j) \setminus A(R_i)$, and $A(R_i) \cap A(R_j)$.
Since $A(R_i) \neq A(R_j)$, either $A(R_i) \setminus A(R_j) \neq \emptyset$ or $A(R_j) \setminus A(R_i) \neq \emptyset$ or both.
We mark some users from $A$ as follows to construct a new authorization relation $A^*$.
If $A(R_i) \setminus A(R_j)\ne \emptyset$, we mark {$\min(\tau, |A(R_i) \setminus A(R_j)|)$ many} users in $A(R_i) \setminus A(R_j)${; in particular if $|A(R_i) \setminus A(R_j)|\le \tau$, then we mark all the users in $A(R_i) \setminus A(R_j)$}. Similarly, for non-empty $A(R_j) \setminus A(R_i)$ 
and $A(R_i) \cap A(R_j)$.
We repeat this process for all pairs $i, j \in [p]$ with $i \neq j$.
We delete all unmarked users from $A$ and output $A^*$ as the new authorization relation.
Clearly, there are at most $3\tau{{k}\choose{2}}$ marked users in $A^*$. Thus, $|A^*| \leq 3\tau k {{k}\choose{2}}$.
Observe that for any marked user $u$, $A^*(u) = A(u)$ and for any $i \in [p]$, for any $r, r' \in R_i$, $A^*(r) = A^*(r')$.

We first argue that for two distinct $i, j \in [p]$, $A^*(R_i) \cap A^*(R_j) \neq \emptyset$ if and only if $A(R_i) \cap A(R_j) \neq \emptyset$.
As $A^* \subseteq A$, if there exists $u \in A^*(R_i) \cap A^*(R_j)$, then the same user $u \in A(R_i) \cap A(R_j)$.
On the other hand, let $A(R_i) \cap A(R_j) \neq \emptyset$.
Then, there exists $u \in A(R_j) \cap A(R_j)$ that we have marked using our marking scheme.
So, $A^*(R_i) \cap A^*(R_j) \neq \emptyset$.

Next we argue that for two distinct $i, j \in [p]$, $A^*(R_i) \neq A^*(R_j)$.
By the definition of $R=R_1 \uplus \ldots \uplus R_p$, we have $A(R_i) \neq A(R_j)$.
Therefore, either there exists $u \in A(R_i) \setminus A(R_j)$ or there exists $u' \in A(R_j) \setminus A(R_i)$ or both.
If there exists $u \in A(R_i) \setminus A(R_j)$, then we have marked {at least} one such $u$.
If there exists $u\in A(R_j) \setminus A(R_i)$, then we have marked  {at least} one such $u$.
As for any marked user $u$, $A^*(u) = A(u)$, $A^*(R_j) \neq A^*(R_i)$.

By the arguments above, we have the following:

\begin{itemize}
	\item $A^* \subseteq A$,
	\item $A(r) = A(r')$ if and only if $A^*(r) = A^*(r')$, and
	\item $A(r) \cap A(r') \neq \emptyset$ if and only if $A^*(r) \cap A^*(r') \neq \emptyset$.
\end{itemize}

Consider a constraint $c = \BoDU{r, r'} \in C$. By~(\ref{w_c:BoDU})~in Section~\ref{sec:vapep}, $w_c(A)=f_c(\maxdiff (A, r, r'))$, where $ \maxdiff (A, r, r')= \max\{|A(r) \setminus A(r')|,|A(r') \setminus A(r)|\}.$
Observe that $\maxdiff (A^*, r, r') \leq \maxdiff (A, r, r')$.
Hence, $w_c(A^*) \leq w_c(A)$.

Consider a constraint $c = \BoDE{r, r'} \in C$. By~(\ref{w_c:BoDE})~in Section~\ref{sec:vapep}, if $A(r) \cap A(r') \neq \emptyset$ then $w_c(A) = 0$; otherwise, $w_c(A) = \ell_c>0$.
As we have proved, $A^*(r) \cap A^*(r') \neq \emptyset$ if and only if $A(r) \cap A(r') \neq \emptyset$.
Hence, $w_c(A) = w_c(A^*) = \ell_c$.

Consider a constraint $c = \SoDU{r, r'} \in C$. By~(\ref{w_c:SoDU})~in Section~\ref{sec:vapep}, $w_c(A)=f_c(|A(r) \cap A(r')|)$.
We have shown that $A^*(r) \cap A^*(r') = \emptyset$ if and only if $A(r) \cap A(r') = \emptyset$.
Moreover, $|A^*(r) \cap A^*(r')| \leq |A(r) \cap A(r')|$.
Hence, $w_c(A^*) \leq w_c(A)$.

Consider a constraint $c = \SoDE{r, r'} \in C$. By~(\ref{w_c:SoDE})~in Section~\ref{sec:vapep}, if $A(r) \neq A(r')$ then $w_c(A) = 0$; otherwise $w_c(A) = \ell_c>0$.
We have proved that $A(r) \neq A(r')$ if and only if $A^*(r) \neq A^*(r')$.
It implies that $w_c(A^*) = w_c(A) = \ell_c$.

Consider a constraint  $c= (r, \leq, t)$. By~(\ref{w_c:card})~in Section~\ref{sec:vapep}, $w_{c}(A)=f_{c}(|A(r)|-t)$.
Observe that $A^*(r) \subseteq A(r)$.
Hence, $w_c(A^*) \leq w_c(A)$. 

Finally, consider a constraint $c= (r, \geq, t)\in C$. By~(\ref{w_c:cardLB})~in Section~\ref{sec:vapep}, $w_{c}(A)=f_{c}(t-|A(r)|)$. Note that \(t\le \tau\). Let \(i,j\in [p]\) be such that \(r\in R_i\) and \(j\neq i\). Notice that \(A(r) = (A(R_i)\setminus A(R_j))\cup (A(R_i)\cap A(R_j))\) and we marked $\min(\tau, |A(R_i) \setminus A(R_j)|)$ users in $A(R_i) \setminus A(R_j)$ and $\min(\tau, |A(R_i) \cap A(R_j)|)$ users in $A(R_i) \cap A(R_j)$. Therefore, if \(|A(r)| \le \tau\), then $A^*(r)=A(r)$ and $w_c(A^*) = w_c(A)$. Otherwise \(|A(r)| \ge \tau\) implying \(|A^*(r)| \ge \tau\ge t\) and $w_c(A^*) = w_c(A) = 0$.

Thus, we conclude that $w_C(A^*) \leq w_C(A)$.
\end{proof}

By definition, if we have user-independent constraints, then we do not need to know which particular users are assigned to resources in order to determine the constraint weight of some authorization relation $A$. 
Instead, it suffices to know for each set $T\subset R$ how many users $u$ are authorized by $A$ precisely for the set $T$, i.e., the size of the set $\{u\in U \mid A(u)=T \}$. 
This leads us to the following definition of the \emph{user profile of an authorization relation}.
Lemma~\ref{lem:user_profile} confirms the intuition behind the definition: if we have a user-independent constraint, then two authorization relations with the same user profile yield the same constraint weight. 

\begin{definition}[\userProfile]
	For a set of resources $R$, a set of users $U$, and an authorization relation $A\subseteq U\times R$, the \emph{\userProfile} of the authorization relation $A$ is the function $\uP_A: 2^R \rightarrow \mathbb{N}$, where $\uPs{A}{T}$ is defined to be $|\{u\in U\mid A(u)=T\}|$.
\end{definition}

Note that $\uP_A(T)$ is not the same as $|A(T)|$. The integer $\uP_A(T)$ is the number of all users that are authorized for all resources in $T$ and nothing else, while $A(T)$ is the set of users that are authorized for at least one resource in $T$. 

An example of a user profile is given in Figure~\ref{fig:user-profile}.
Here $R = \{r_1, r_2, r_3, r_4\}$ and $U = \{u_1, u_2, u_3, u_4, u_5\}$.
Let $A$ be an authorization relation, shown in Figure~\ref{fig:authorization-relation} in the form of a bipartite graph, such that $A(r_1) = \{u_1, u_2, u_3\}$, $A(r_2) = \{u_2, u_3, u_4\}$, $A(r_3) = \{u_4\}$ and $A(r_4) = \{u_4\}$.
Then Figure~\ref{fig:user-profile} shows the user profile of $A$.

\begin{figure}
\begin{subfigure}[b]{0.45\textwidth}
\begin{center}
\begin{tikzpicture}[
	edge/.style={draw=black}
]
\foreach \i in {1, ..., 4}
{
	\pgfmathtruncatemacro{\y}{4 - \i};
	\node[circle, minimum size=1em, draw=black] (r\i) at (0, \y) {$r_{\i}$};
}

\foreach \i in {1, ..., 5}
{
	\pgfmathsetmacro{\y}{4.5 - \i};
	\node[circle, minimum size=1em, draw=black] (u\i) at (2, \y) {$u_{\i}$};
}

\draw[edge] (r1)--(u1);
\draw[edge] (r1)--(u2);
\draw[edge] (r1)--(u3);
\draw[edge] (r2)--(u2);
\draw[edge] (r2)--(u3);
\draw[edge] (r2)--(u4);
\draw[edge] (r3)--(u4);
\draw[edge] (r4)--(u4);
\end{tikzpicture}
\end{center}
\caption{Authorization relation}
\label{fig:authorization-relation}
\end{subfigure}
\hfill
\begin{subfigure}[b]{0.45\textwidth}
\begin{center}
\begin{tabular}{lc}
\toprule
$T$ & $\uP_A(T)$\\
\midrule
$\{r_1\}$ & 1 \\
$\{r_1, r_2\}$ & 2 \\
$\{r_2, r_3, r_4\}$ & 1 \\
$\emptyset$ & 1 \\
Other $T \subseteq R$ & 0 \\
\bottomrule
\end{tabular}
\end{center}
\vspace{5ex}
\caption{User profile}
\label{fig:user-profile}
\end{subfigure}

\caption{
	An example of an authorization relation and the corresponding user profile.
}
\end{figure}


\begin{lemma}\label{lem:user_profile}
	Let $U$ be a set of users, $R$ a set of resources, $(c, w_c)$ a user-independent weighted constraint and $A_1,A_2\subseteq U\times R$ two authorization relations such that $\uPs{A_1}{T}=\uPs{A_2}{T}$ for all $T\subseteq R$. Then $w_c(A_1)=w_c(A_2)$. 
\end{lemma}
\begin{proof}
	We will define a permutation $\sigma: U \rightarrow U$ such that $\sigma(A_1)=A_2$. The lemma then immediately follows from the definition of user-independence.
	For $i\in \{1,2\}$ and $T\subseteq R$, let $U^i_T$ be the set of users that are assigned by $A_i$ precisely to the resources in $T$ and nothing else. That is 
	\(U^i_T = \{u\in U\mid A_i(u)=T\}. \)
	Now, let us fix for each $U^i_T$ an arbitrary ordering of the users in $U^i_T$ and let $u^i_{T,j}$ for $j\in [|U^i_T|]$ denote the $j$-th user in $U^i_T$. Note that for all $T\subseteq R$, we have $\uPs{A_1}{T}=\uPs{A_2}{T}$ by the assumptions of the lemma and hence $|U^1_T| = |U^2_T|$. Moreover, each user in $U$ is assigned exactly one (possibly empty) subset of resources in each of the authorization relations $A_1$ and $A_2$. Hence the sets $\bigcup_{T\subseteq R}\{U^1_T\}$ and $\bigcup_{T\subseteq R}\{U^2_T\}$ are both partitions of $U$.
	We are now ready to define the permutation $\sigma$ as $\sigma(u^1_{T,j}) = u^2_{T,j}$ for all $T\subseteq R$, $j\in [|U^1_T|]$. It remains to show that $\sigma(A_1)=A_2$. By the definition of the users $u^1_{T,j}$ and $u^2_{T,j}$, we get that for all $T\subseteq R$, all $j\in [|U^1_T|]$, and all $r\in R$ we have $(u^1_{T,j},r)\in A_1$ if and only if $r\in T$ if and only if $(u^2_{T,j},r)\in A_2$ and the lemma follows.  
\end{proof}

\begin{lemma}\label{lem:computation_of_small_solution}
	Let $\cI = (R, U, C, \omega)$ be an instance of {\vapep} such that all constraints in $C$ are user-independent and let $\uP: 2^R \rightarrow \mathbb{N}$ be a \userProfile. Then there exists an algorithm that finds a relation $A$ which minimizes $w(A)$ among all relations with user profile $\uP$.
\end{lemma}

\begin{proof}
	It follows from Lemma~\ref{lem:user_profile} and the fact that all constraints in $C$ are user-independent that $w_C(A)$ only depends on the \userProfile\ of $A$ and hence we only need to find an authorization relation $A$ with \userProfile\ $\uP_A=\uP$ such that $\Omega(A)$ is minimized.
	
	Note that if $\sum_{T\subseteq R}\uP(T) \neq |U|$, then there is no authorization relation with given user profile. This is because for every authorization relation $A$ and every user $u$, the set $A(u)$ is defined as is a (possibly empty) subset of $R$. Hence from now on we assume that $\sum_{T\subseteq R}\uP(T) = |U|$. We start by creating a weighted complete bipartite graph $G=(V_1\cup V_2, E)$, with parts $V_1, V_2$ such that $V_1=U$ and $V_2$ contains for each $T\subseteq R$ a set of $\uP(T)$ vertices; let us denote these vertices $\{v^T_1, v^T_2,\ldots, v^T_{\uP(T)}\}$. For a user $u\in U$ and a vertex $v^T_i$, the weight of the edge $uv^T_i$ is defined as $w(uv^T_i)=\omega(u,T)$. Since  $\sum_{T\subseteq R}\uP(T) = |U|$, it follows that $|V_1| = |V_2|$. We show that there is a correspondence between perfect matchings of the graph $G$ and authorization relations with \userProfile\ $\uP$. 
	
	First, let $A$ be an authorization relation such that $\uP_A = \uP$. Then, we can get a perfect matching $M_A$ of $G$ of weight $\Omega(A)$ as follows. Because, $\uP_A = \uP$, we have that for every $T\subseteq R$ there are exactly $\uP(T)$ many users $u\in U$ such that $A(u)=T$. Hence for every $T\subseteq R$ there is a perfect matching $M_A^T$ between these users and vertices $\{v^T_1, v^T_2,\ldots, v^T_{\uP(T)}\}$. Moreover, the cost of an edge between a user $u\in U$ such that $A(u)=T$ and a vertex $v^T_j$, $j\in \uP(T)$, is $\omega(u,T)$ which is precisely the contribution of the user $u$ to $\Omega(A)$. Hence the cost of the matching $M_A = \bigcup_{T\subseteq R}M_A^T$ is precisely $\Omega(A)$. 
	
	On the other hand if $M$ is a perfect matching in $G$, then we can define an authorization relation $A_M$ as $(u,r)\in A_M$, if and only if $u$ is matched to a vertex $v^T_i$ with $r\in T$. Clearly, every user $u$ is then matched by $M$ to a vertex $v^T_i$ such that $A_M(u)=T$ and weight of the edge in $M$ incident to $u$ is precisely $\omega(u, A_M(u))$, which is the contribution of $u$ to $\Omega(A)$.
	
	It follows that $G$ has a perfect matching of cost $W$ if and only if there is an authorization relation $A$ with \userProfile\ $\uP$ and $\Omega(A)=W$ and given a perfect matching of $G$, we can easily find such an authorization relation. Therefore, to finish the proof of the lemma we only need to compute a minimum cost perfect matching in the weighted bipartite graph $G$, which can be done using the well-known Hungarian method in $\cO(mn)$ time~\cite{kuhn1956variants}, where $n$ is the number of vertices and $m$ is the number of edges in $G$.
\end{proof}

The next ingredient required to prove our main result (Theorem~\ref{thm:t_bounded_algorithm}) is the fact that the number of all possible user profiles for all authorization relations of size at most $t$ is small and can be efficiently enumerated. 
Let $\cI = (R, U, C, \omega)$ be an instance of {\vapep} and $A\subseteq U\times R$ an authorization relation. 
Then for the \userProfile\ $\uP_A$ of $A$ we have that $\sum_{T\subseteq R} |T|\cdot \uP_{A}(T)= |A|$ and $\sum_{T\subseteq R} \uP_{A}(T) = |U|$. 
Moreover, if $A$ is complete, then $t\ge k$. Note that the number of users in an optimal solution for a $t$-wbounded set of weighted constraints is at most $t$ by Lemma~\ref{lem:existance_of_small_solution}. However, sometimes we are able to show that the number of users in an optimal solution is actually significantly smaller than the bound $t$ such that the set of weighted constraints $C$ is $t$-wbounded (see, e.g.,  { Lemma~\ref{lemma:quadratic-w-boundedness-bod-sod}}). Moreover, if $\uP_A$ is a \userProfile\ of an authorization relation with at most $\ell$ users, then $\sum_{T\subseteq R, T\neq \emptyset} \uP_{A}(T) \le \ell$.
The following lemma will be useful because we are only interested in complete authorization relations of size at most $t$ that use at most $\ell\le t$ users.

\begin{lemma}\label{lem:NrUsrProfiles}
	Let $\cI = (R, U, C, \omega)$ be an instance of {\vapep} such that $|R|=k$ and let $\ell\in \mathbb{N}$. 
	Then the number of possible \userProfile{s}, $\uP: 2^R\rightarrow \mathbb{N}$, such that $\sum_{T\subseteq R, T\neq \emptyset}\uP(T) \le \ell$
	is $\binom{\ell+2^k-1}{\ell}$.
	Moreover, we can enumerate all such functions in time $\cO^*(\binom{\ell+2^k-1}{\ell})$.
\end{lemma}

\begin{proof}
	It is well known that the number of weak compositions of a natural number $q$ into $p$ parts (the number of ways we can assign non-negative integers to the variables $x_1, x_2, \ldots, x_p$ such that $\sum_{i=1}^p x_i = q$) is precisely $\binom{p+q-1}{q-1}= \binom{p+q-1}{p-1}$ (see, e.g.,~\cite{Jukna01}). Note that because $\sum_{T\subseteq R} \uP(T) = |U|$, each user profile $\uP$ is determined by assigning $\uP(T)$ for all $T\neq \emptyset$. It is not difficult to see that the number of ways in which we can assign $\uP(T)$ for all $T\neq \emptyset$ such that $\sum_{T\subseteq R, T\neq \emptyset}\uP(T) \le \ell$ is the same as the number of weak partitions of $\ell$ into $2^k$ parts - each of the first $2^k-1$ parts is identified with one of $2^k-1$ sets $T\subseteq R$ such that $T\neq \emptyset$. The last part is then a ``slack'' part that allows $\sum_{T\subseteq R, T\neq \emptyset}\uP(T)$ to be also smaller than $\ell$. 
	It follows that the number of possible user profiles is at most $\binom{\ell+2^k-1}{\ell}$. To enumerate them in $\cO^*(\binom{\ell+2^k-1}{\ell})$ time we can do the following branching algorithm: We fix some order $T_1, T_2, \ldots, T_{2^k-1}$ of the non-empty subsets of $R$. We first branch on $\ell+1$ possibilities for $\uP(T_1)$, then we branch on $\ell+1- \uP(T_1)$ possibilities for $\uP(T_2)$, and so on, until we branch on $\ell+1- \sum_{i\in [2^k-2]}\uP(T_i)$ possibilities for $\uP(T_{2^k-1})$. Afterwards, we compute $\uP(\emptyset)$ from $\sum_{T\subseteq R} \uP(T) = |U|$. Each leaf of the branching tree gives us a different possible \userProfile\ and we spend polynomial time in each branch. Hence the running time of the enumeration algorithm is $\cO^*(\binom{\ell+2^k-1}{\ell})$.	
\end{proof}
Because the number of possible \userProfile{s} that authorize at most $\ell$ users appears in the running time of our algorithms, 
it will be useful to keep in mind the following two simple observations about the combinatorial number $\binom{\ell+2^k-1}{\ell}$.
\begin{observation}\label{obs:binom_bound1}
	$\binom{\ell+2^k-1}{\ell} \le 2^{\ell + 2^k-1}\le 4^{\max(\ell, 2^k-1)}$.
\end{observation}
\begin{observation}\label{obs:binom_bound2}
	If $\ell\ge 4$, then $\binom{\ell+2^k-1}{\ell} \le \min(2^{\ell k}, \ell^{2^k-1}) +1$.
\end{observation}
\begin{proof}
	If $k = 0$, then $\binom{\ell+2^k-1}{\ell}=\binom{\ell}{\ell}=1\le \min(2^{\ell\cdot 0}, \ell^{2^0-1}) +1$. 
	If $k = 1$, then $\binom{\ell+2^k-1}{\ell} = \ell +1\le \min(2^{\ell\cdot 1}, \ell^{2^1-1}) +1$. 
	If $k = 2$, then $\binom{\ell+2^k-1}{\ell} = 
	\binom{\ell+3}{3} = \frac{\ell^3+6\ell^2+11\ell + 1}{6}$ and since $\ell \ge 4$, it follows that $\binom{\ell+3}{3}\le \ell^3\le 2^{2\ell}$. 
	From now on, let us assume that $k\ge 3$ and $\ell \ge 4$.
	We distinguish between two cases depending on whether $\ell< 2^k$ or $\ell\ge 2^k$.
	Let us first consider $\ell\ge 2^k$. Note that in this case $\ell^{2^k-1}\le 2^{\ell k}$ and because $k\ge 2$ we have
	$\binom{\ell+2^k-1}{\ell} = \binom{\ell+2^k-1}{2^k-1} \le \frac{(\ell+2^k-1)^{2^k-1}}{(2^k-1)!}\le \frac{2^{2^k-1}}{(2^k-1)!}\cdot \ell^{2^k-1}\le \ell^{2^k-1}$.
	On the other hand, let us now assume $\ell \le 2^k-1$. Note that, because $k\ge 2$ and $\ell\ge 3$, it holds that $(2^k-1)^{\ell} \le \ell^{2^k-1}$. Furthermore, $(2^k-1)^{\ell}< 2^{\ell k}$. 
	Then $\binom{\ell+2^k-1}{\ell} \le \frac{(\ell+2^k-1)^\ell}{\ell!}\le \frac{2^\ell}{\ell!}\cdot (2^k-1)^{\ell}\le (2^k-1)^{\ell}.$
\end{proof}
We are now ready to state and prove the main lemma of this section, which establishes that there exists an FPT algorithm that finds the best solution among all solutions that 
authorize at most $\ell$ users. In particular, this algorithm finds an optimal solution for the case of user-independent $t$-wbounded constraints.

\begin{lemma}\label{lem:t_bounded_algorithm_users}
	Let $\cI = (R, U, C, \omega)$ be an instance of {\vapep} such that all weighted constraints in $C$ are user-independent and let $\ell\in \mathbb{N}$.
	Then there exists an algorithm
	that in time $\cO^*(\binom{\ell+2^k-1}{\ell})$ computes a complete authorization relation $A$ such that $w(A) \leq w(A')$ for every complete authorization relation $A' \subseteq U \times R$ that authorizes at most $\ell$ users for some resource in $R$.
\end{lemma}
\begin{proof} 
	Note that it suffices to compute such an authorization relation $A$ that also authorizes at most $\ell$ users. Let $A^*$ be one such complete authorization relation that satisfies the statement of the theorem.
	Let $\uP_{A^*}: 2^R\rightarrow \mathbb{N}$ be the user profile of $A^*$. Observe that $\sum_{T\subseteq R, T\neq \emptyset} \uP_{A^*}(T)\le \ell$. Moreover, observe that $\sum_{T\subseteq R} \uP_{A^*}(T) = |U|$, as every user in $U$ is assigned to precisely one subset of resources by $A^*$. 
	Furthermore, notice that since $A^*$ is complete, for all $r\in R$ we have that $\sum_{\{r\}\subseteq T\subseteq R}\uP_{A^*}(T)\ge 1$ and we may restrict our attention to \userProfile{s} that also satisfy $\sum_{\{r\}\subseteq T\subseteq R}\uP_{A^*}(T)\ge 1$ for all $r\in R$.
	By Lemma~\ref{lem:NrUsrProfiles}, there exist $\binom{\ell+2^k-1}{\ell}$ different functions (possible \userProfile{s}) $\uP: 2^R\rightarrow \mathbb{N}$ such that $\sum_{T\subseteq R, T\neq \emptyset} \uP(T)\le \ell$ and $\sum_{T\subseteq R} \uP(T) = |U|$. Moreover, we can enumerate all of them in time $\cO^*(\binom{\ell+2^k-1}{\ell})$.
	
	Now, let $\uProfSet$ be the set of all such possible user profiles obtained by Lemma~\ref{lem:NrUsrProfiles} that also satisfy $\sum_{\{r\}\subseteq T\subseteq R}\uP_{A^*}(T)\ge 1$ for all $r\in R$.	
	Since $\uProfSet$ is a subset of functions computed by Lemma~\ref{lem:NrUsrProfiles}, it follows that $|\uProfSet|\le \binom{\ell+2^k-1}{\ell}$. 
	Moreover, it is easy to see that $\uP_{A^*}\in \uProfSet$. 
	The algorithm then branches on all possible profiles in $\uProfSet$ and for a profile $\uP_i\in \uProfSet$, $i\in [|\uProfSet|]$, it computes an authorization relation $A_{i}$ such that $\uP_{A_i}=\uP_i$ and $w(A_i)$ is minimized, which can be done in polynomial time by Lemma~\ref{lem:computation_of_small_solution}. 
	Finally, the algorithm outputs the authorization relation $A_i$ for the \userProfile\ $\uP_i$ that minimizes $w(A_i)$ among all $\uP_i\in \uProfSet$. 
	The running time of the whole algorithm is $\cO^*(\binom{\ell+2^k-1}{\ell})$. 
	
	To establish correctness, first notice that for all $i\in [|\uProfSet|]$ we have $\sum_{\{r\}\subseteq T\subseteq R}\uP_{i}(T)\ge 1$ for all $r\in R$, so the authorization relation $A_i$ is complete. 
	Furthermore, recall that $\uP_{A^*}\in \uProfSet$. 
	For $i\in[|\uProfSet|]$ such that $\uP_i = \uP_{A^*}$, we have that $w(A_i)\le w(A^*)\leq w(A')$ for all complete authorization relations $A' \subseteq U \times R$ that authorizes at most $\ell$ users for some resource in $R$.
\end{proof}
Note that it follows from Lemma~\ref{lem:existance_of_small_solution} that given an instance $\cI= (R, U, C, \omega)$ of {\vapep} such that all weighted constraints in $C$ are user-independent and $C$ is $t$-wbounded there exists an optimal solution $A^*$ of $\cI$ such that $|A^*|\le t$. Moreover, $|A^*|\le t$ implies that $A^*$ authorizes at most $t$ users for some resource. Hence, in combination with Observation~\ref{obs:binom_bound2}, we immediately obtain the main result of this section as a corollary, which establishes that there exists an FPT algorithm for the case of user-independent $t$-wbounded constraints.  

\begin{theorem}\label{thm:t_bounded_algorithm}
	Let $\cI = (R, U, C, \omega)$ be an instance of {\vapep} such that all weighted constraints in $C$ are user-independent and $C$ is $t$-wbounded. 
	Then there exists an algorithm solving $\cI$ in time $\cO^*(\binom{t+2^k-1}{t})= \cO^*(\min(2^{tk}, t^{2^k-1})) = \cO^*(2^{\min(kt, ({2^k-1})\log t)})$.
\end{theorem}

Using Theorem~\ref{thm:t_bounded_algorithm} and Lemma~\ref{lemma:quadratic-w-boundedness-bod-sod}, we have the following:

\begin{corollary}\label{cor:FPT_known_constraints}
\label{thm:improved-result}
Let $\tau = \max_{(r,\ge,t)\in C}t$. Then {\vapep}$\ev{\bodU, \bodE, \sodE, \sodU, {\sf Card_{UB}}, {\sf Card_{LB}}}$
 can be solved in $\cO^*(8^{\tau k^2 {{k}\choose{2}}})$ time. Thus, {\vapep}$\ev{\bodU, \bodE, \sodE, \sodU, {\sf Card_{UB}}, {\sf Card_{LB}}}$ parameterized by $k+\tau$ is FPT.
\end{corollary}

\begin{proof}
As $t \le \tau 3k{{k}\choose{2}}$, we have $2^{kt} \le 8^{k^2{{k}\choose{2}}}$ completing the proof.
\end{proof}

In Sections~\ref{sec:U}~and~\ref{sec:vapep-bodE-sodU} we develop more specialized algorithms that solve \vapep\ more efficiently for instances where all constraints are from some specific subset of the above constraints. 
We conclude this section by showing that restricting attention to user-independent constraints is not sufficient to obtain an FPT algorithm parameterized by the number of resources $k$ even for APEP. 
Because all weighted constraints are necessarily $(k\cdot |U|)$-wbounded, as a corollary of our conditional lower-bound, we will also show that, unless the Exponential Time Hypothesis (ETH) fails, the algorithm in Theorem~\ref{thm:t_bounded_algorithm} is in a sense best that we can hope for from an algorithm that can solve \vapep\ with arbitrary user-independent $t$-wbounded weighted constraints. 

Because all the constraints we saw so far were $2^k$-wbounded, we will need to introduce new, more restrictive, constraints to obtain our W[2]-hardness and ETH lower-bound. As we consider only user-independent constraints, it is natural for a constraint $c$ 
to be a function of the \userProfile\ of $A$. To this end, we define a new type of user-independent constraint $c = (\tau, X, \bigvee)$, where $\tau\subseteq 2^R$ and $X\subseteq \mathbb{N}$. The constraint $(\tau, X, \bigvee)$ is satisfied if and only if there exist $T\in \tau$ and $x\in X$ such that $\uPs{A}{T} = x$.
Less formally, $c$ is satisfied if and only if some specified number of users (in $X$) is authorized for some specified set of resources (in $\tau$).
To obtain our ETH lower-bound we make use of the following well-known result in parameterized complexity:

\begin{theorem}[\cite{LokshtanovMS11}]\label{thm:DomSetETH}
	Assuming ETH, there is no $f(k)n^{o(k)}$-time algorithm for \textsc{Dominating Set}, where $n$ is the number of the vertices of the input graph, $k$ is the size of the output set, and $f$ is an arbitrary computable function.
\end{theorem}

Given the above theorem, we are ready to prove the main negative result of this section.

\begin{theorem}\label{mainneg}
	APEP is W[2]-hard and, assuming ETH, there is no $f(|R|)\cdot |\cI|^{o(2^{|R|})}$-time algorithm solving APEP even when all constraints are user-independent and the base authorization relation is $U\times R$. 
\end{theorem}

Henceforth, we will write $[k]$ to denote $\{1,\dots,k\}$.
For a graph $G = (E,V)$ and vertex $x\in V$, $N_G(x) = \{y \in V \mid xy \in E\}$ is the set of vertices adjacent to $x$ in $G$; for a set $S\subseteq V(G)$, $N(S)=\bigcup_{x\in S} N_G(x) \setminus S.$

\begin{proof}
	To prove the theorem we give a reduction from the \textsc{Dominating Set} problem. 
	Let $(G,k)$ be an instance of the \textsc{Dominating Set} problem. 
	Let $|V(G)|= n$ and let $V(G)= \{v_1, v_2, \ldots, v_n\}$; that is we fix some arbitrary ordering of the vertices in $G$ and each vertex of $G$ is uniquely identified by its position in this ordering (index of the vertex). 
	For a vertex $v_i\in V(G)$ we let the set $X_i = \{i\}\cup \{j \mid v_j\in N_G(v_i) \}$. 
	In other words, for a vertex $v_i\in V(G)$, the set $X_i$ is the set of indices of the vertices in the closed neighbourhood of $v_i$. The aim of the \textsc{Dominating Set} problem is then to decide whether $G$ has a set $S$ of at most $k$ vertices such that for all $i\in [n]$ the set $X_i$ contains an index of some vertex in $S$.
	
	Let $\cI = (U, R, \hA, C)$ be an instance of APEP such that
	\begin{itemize}
		\item $R=\{r_1, \ldots, r_\ell\}$ such that $2^{\ell-1}\le k < 2^{\ell}$,
		\item $|U|= k\cdot n$,
		\item $C = \bigcup_{i\in [n]}\{(\tau, X_i, \bigvee)\}$, where $\tau\subset 2^R$ such that $\emptyset\notin \tau$ and $|\tau| =k$, and 
		\item $\hA = U\times R$. 
	\end{itemize} 
	
	We prove that $(G,k)$ is a YES-instance of \textsc{Dominating Set} if and only if $\cI$ is a YES-instance of APEP. 
	Let $\tau = \{T_1, T_2, \ldots, T_k\}$. Observe that because $2^{\ell-1}\le k$ and $T_i\neq T_j$ for $i\neq j$, it follows that every resource appears in $T_i$ for some $i\in [k]$. 
	
	Let $S = \{v_{q_1}, v_{q_2},\ldots,v_{q_k}\}$ be a dominating set of $G$ of size $k$ (note that if we have a dominating set of size at most $k$, then we have a dominating set of size exactly $k$). Let $A$ be an authorization relation such that $\uPs{A}{T_i}=q_i$. Because $|U|= k\cdot n$ and $1\le q_i\le n$ for all $i\in [k]$, it is easy to construct such an authorization relation. For each $i\in [k]$ we simply select $q_i$ many fresh users $u$ such that $A(u)=T_i$ and leave the remaining users not assigned to any resources ($A(u)=\emptyset$).
	Because $2^{\ell-1}\le k$ and for all $T_i\in \tau$ we have $\uP_A(T_i)\ge 1$, it is easy to see that $A$ is a complete authorization relation. Since $\hA = U\times R$, $A$ is authorized. 
	It remains to show that $A$ is eligible w.r.t. $C$. Consider the constraint $c_i = (\tau, X_i, \bigvee)$. Since $S$ is a dominating set, the closed neighbourhood of $v_i$ contains a vertex $v_{q_j}\in S$. But then $q_j\in X_i$, $T_j\in \tau$, and $\uPs{A}{T_j}=q_j$, hence $c_i$ is satisfied. 
	
	On the other hand let $A$ be valid w.r.t. $\hA$. We obtain a dominating set of $G$ of size at most $k$ as follows.
	Without loss of generality, let as assume that if $i<j$, then $\uPs{A}{T_i}\ge \uPs{A}{T_j}$ and let $k'\in [k]$ be such that $\uPs{A}{T_{k'}}\ge 1$ and $\uPs{A}{T_{k'+1}}= 0$ (note that if $\uPs{A}{T_{k}}\ge 1$, then $k'=k$). We let $S = \bigcup_{i\in [k']} \{v_{\uPs{A}{T_i}}\}$. We claim that $S$ is a dominating set (clearly $|S|=k'\le k$). Let $v_i$ be arbitrary vertex in $V(G)\setminus S$. Consider the constraint $c_i = (\tau, X_i, \bigvee)$. Clearly $c_i$ is satisfied and there exists $T_j\in \tau$ and $x\in X_i$ such that $\uPs{A}{T_j}=x$. But by definition of $X_i$, $x\ge 1$ and $v_x$ is a neighbour of $v_i$. Moreover, by the definition of $S$, we have $v_x = v_{\uPs{A}{T_j}}\in S$. It follows that $S$ is a dominating set. 
	
	Now for each $i\in [n]$, the set $X_i$ has size at most $n$ and $\tau$ has size $k\le n$, so the size of the instance $\cI$ is polynomial in $n$. Moreover $2^{\ell-1}\le k < 2^{\ell}$, hence {APEP} is W[2]-hard parameterized by $|R|$ and an $f(|R|)\cdot |\cI|^{o(2^{|R|})}$ time algorithm for APEP yields an $f(k)n^{o(k)}$ time algorithm for \textsc{Dominating Set}, and the result follows from Theorem~\ref{thm:DomSetETH}.
\end{proof}



The set $C$ is trivially $(|R|\cdot |U|)$-wbounded for every \vapep\ instance $\cI = (R,U, C, \omega)$, so we obtain the following result, which asserts that the lower bound asymptotically matches the running time of the algorithm from Theorem~\ref{thm:t_bounded_algorithm}. 
\begin{corollary}
	Assuming ETH, there is no $t^{o(2^{|R|})}\cdot n^{\cO(1)}$ time algorithm that given an instance $\cI =(R, U, C, \omega)$ of \vapep\ such that all constraints in $C$ are user-independent and $C$ is $t$-wbounded computes an optimal solution for $\cI$. 
\end{corollary}

\section{${\sodU}$ and ${\bodU}$ Constraints}\label{sec:U}

In this section, we will consider {\sc Valued APEP}, where all constraints are only ${\bodU}$ and  ${\sodU}$. 
We will show how to reduce it to {\sc Valued WSP} with user-independent constraints, with the number of steps equal to the number $k$ of resources in {\sc Valued APEP}. 
As a result, we will be able to obtain an algorithm for {\sc Valued APEP} with only ${\bodU}$ and  ${\sodU}$ constraints of running time   $\cO^*(2^{k\log k})$.

Let us start with {\sc Valued APEP}$\ev{\sodU}$. 
Recall that the weighted version of an ${\sodU}$ constraint $(r,r',\updownarrow,\forall)$ is $w_c(A) = f_c(|A(r) \cap A(r')|)$ for some monotonically increasing function $f_c$. 

The weight of a binary SoD constraint $c=(s',s'',\ne)$ in {\sc Valued WSP} is $0$ if and only if 
steps $s'$ and $s''$ are assigned to different users. 
{\sc Valued WSP} using only SoD constraints of this form will be denoted by {\sc Valued WSP}($\ne$). 

\begin{lemma}\label{lemma:only-sod-universal} \sloppy
Let ${\cI}=(R, U, C, \omega)$ be an instance of $\mbox{ }$ {\sc Valued APEP}$\ev{\sodU}$ and let $A^*$ be an optimal solution of $\cI$. 
Let $A'$  be arbitrary authorization relation such that $A'\subseteq A^*$ and $|A'(r)|=1$ for every $r\in R$. Then $A'$ is an optimal solution of $\cI$. Moreover, in polynomial time $\cI$ 
can be reduced to an instance ${\cI}'$ of {\sc Valued WSP}$(\ne)$ such that the weights of optimal solutions of $\cI$ and ${\cI}'$ are equal.
\end{lemma}

\begin{proof}
Let $A'$ be an arbitrary relation such that $A'\subseteq A^*$ and $|A'(r)|=1$ for every $r\in R$. 
By definition, $A'$ is complete.
By (\ref{moncond}) and \eqref{eq:Omega}, $A'\subseteq A^*$ implies $\Omega(A') \le \Omega(A^*).$ 
Let $c = \SoDU{r, r'}\in C$. 
Since $A'(r)\cap A'(r')\subseteq A^*(r)\cap A^*(r')$, $w_c(A) = f_c(|A(r) \cap A(r')|)$ and $f_c$ is non-decreasing, we have $w_c(A')\le w_c(A^*)$. 
Thus, $\Omega(A')+w_C(A') \le\Omega(A^*)+w_C(A^*).$ Since $A^*$ is optimal, $A'$ is optimal, too. 

Define an instance ${\cI}'$ of {\sc Valued WSP}($\ne$) as follows: the set of steps is $R$, the set of users is $U$, and $(r,r',\ne)$ is a constraint of ${\cI}'$ if  $\SoDU{r, r'}$ is a constraint of $\cI$. 
The weight of $(r,r',\ne)$ equals $w_c(A')=f_c(1)$ (recall that $|A'(r)| = 1$ for all $r$) and the weights of $(u,T)$, $u\in U, T\subseteq R$, in both ${\cI}$ and ${\cI}'$ are equal.
Observe that $\pi:\ R\to U$ defined by $\pi(r)=A'(r)$ is an optimal plan of ${\cI}'$. 
Thus, the optimal solution of $\cI$ has the same weight as that of ${\cI}'$.
\end{proof}

We now will consider {\sc Valued APEP}$\ev{\bodU,\sodU}$. 
Recall that the weighted version of a ${\bodU}$ constraint $(r,r',\leftrightarrow,\forall)$ is given by $f_c(\max\{|A(r) \setminus A(r')|,|A(r') \setminus A(r)|\})$ for some monotonically increasing function $f_c$.

The weight of binary BoD constraint $c=(s',s'',=)$ in {\sc Valued WSP} is  $0$ if and only if  steps $s'$ and $s''$ are assigned the same user. 
{\sc Valued WSP}$(=)$ denotes {\sc Valued WSP} containing only BoD constraints; {\sc Valued WSP}$(=,\ne)$ denotes {\sc Valued WSP} containing only SoD and BoD constraints.
In fact, {\sc Valued WSP}($=$) is already {\sf NP}-hard which follows from Theorem 6.4 of \cite{CohenCJK05}. 
This theorem, in particular, shows that {\sc Valued WSP}($=$) is {\sf NP}-hard even if the weights are restricted as follows: $w_c(\pi)=1$ if a plan $\pi$ falsifies a constraint $c$, $\omega(u,r)=\infty$ if $(u,r)\not\in \hA.$  

\begin{lemma}
\label{lemma:only-bod-sod-universal}\sloppy
Let ${\cI}=(R, U, C,\omega)$ be an instance of {\sc Valued APEP}$\ev{\bodU,\sodU}$ and let $A^*$ be an optimal solution of ${\cI}$.
There is an optimal solution $A'$ of ${\cI}$ such that $A'\subseteq A^*$ and $|A'(r)|=1$ for every $r\in R$. Moreover, in polynomial time $\cI$ 
can be reduced to an instance ${\cI}'$ of \vwsp$(=,\ne)$. 
\end{lemma}

\begin{proof} 
Consider an optimal solution $A^* $ of $\cI$ and define $A'$ as follows. 
We first define an equivalence relation $\cong$ on $R$, where $r \cong r'$ if and only if $A^*(r) = A^*(r')$.
This gives a partition $R = R_1 \uplus \ldots \uplus R_p$ such that $p \leq k$. For each $R_i$, we choose $u_i\in A^*(r)$, where $r\in R_i.$
Then $A'=\cup_{i=1}^p \{(u_i,r):\ r\in R_i\}.$ 

By definition, $A'$ is complete. By (\ref{moncond}) and \eqref{eq:Omega},  $A'\subseteq A^*$ implies $\Omega(A') \le \Omega(A^*).$ 
Let $c = \BoDU{r, r'}\in C$. If $A^*$ satisfies $c$ then  $A'$ also satisfies $c.$ 
If $c$ is falsified by $A^*$ then $$ \max\{|A^*(r) \setminus A^*(r')|, |A^*(r') \setminus A^*(r)|\} \ge 1$$ but $ \max\{|A'(r) \setminus A'(r')|, |A'(r') \setminus A'(r)|\}= 1.$ Hence, $w_c(A^*) \ge f_c(1)=w_{c}(A')$. 
Now let $c = \SoDU{r, r'}\in C$. 
By the proof of Lemma \ref{lemma:only-sod-universal}, we have $w_c(A^*) \ge w_{c}(A').$

Thus, $\Omega(A')+w_C(A') \le \Omega(A^*)+w_C(A^*).$ Since $A^*$ is optimal, $A'$ is optimal, too. 
  
An instance of  ${\cI}'$ of {\sc Valued WSP}($=,\ne$) is defined as in Lemma~\ref{lemma:only-sod-universal}, but the constraints $(r,r',=)$ correspond to constraints \mbox{$ \BoDU{r, r'}$} in $C.$ 
It is easy to see that the optimal solution of $\cI$ has the same weight as that of ${\cI}'$.
\end{proof}

We are now able to state the main result of this section.
The result improves considerably on the running time for an algorithm that solves {\sc Valued APEP} for arbitrary weighted $t$-bounded user-independent constraints (established in Theorem~\ref{thm:t_bounded_algorithm}).

\begin{theorem}
\label{thm:bod-sod-universal-to-WSP-equality}
{\vapep}$\ev{\bodU,\sodU}$ is FPT and can be solved in time $\cO^*(2^{k\log k})$.
\end{theorem}

\begin{proof}
Let $\cI$ be an instance of {\sc Valued APEP}$\ev{\bodU,\sodU}$. 
By Lemma \ref{lemma:only-bod-sod-universal}, $\cI$ can be reduced to an instance 
${\cI}'$ of {\sc Valued WSP}($=,\ne$). 
It remains to observe that ${\cI}'$ can be solved in time $\cO^*(2^{k\log k})$ using the algorithm of Theorem~\ref{thm1}, as $(r,r',=)$ and $(r,r',\ne)$ are user-independent constraints.
\end{proof}

\section{${\bodE}$ and ${\sodU}$ constraints}
\label{sec:vapep-bodE-sodU}

In this section, we consider {\vapep}$\ev{\bodE, \sodU}$.
We provide a construction that enables us to reduce an instance $\cI$ of {\vapep}$\ev{\bodE, \sodU}$ with $k$ resources to an instance $\cI'$ of {\vwsp} with only user-independent constraints containing at most $k(k-1)$ steps.
Moreover, the construction yields a {\sc Valued WSP} instance in which the weight of an optimal plan is equal to the weight of an optimal solution for the {\sc Valued APEP} instance.
Finally, we show that it is possible to construct the optimal solution for the {\sc Valued APEP} instance from an optimal plan for the {\sc Valued WSP} instance.

\sloppy Let $\cI = (R, U, C, \omega)$ be an instance of {\vapep}$\ev{\bodE, \sodU}$.
(The weights of these types of constraints are defined by equations~\eqref{w_c:BoDE} and~\eqref{w_c:SoDU} in Section~\ref{sec:apep-constraints}.)
Let $R = \{r_1,\ldots,r_k\}$.
Then we construct an instance $\cI' = (S', U', C', \omega')$ of {\vwsp} as follows.

\begin{itemize}
	\item Set $U' = U$.
	\item For every $i \in [k]$, first initialize a set $\Gamma(r_i) = \emptyset$.
	Then, for every $\bodE$ constraint $\BoDE{r_i, r_j} \in C$, add $r_j$ to $\Gamma(r_i)$ and $r_i$ to $\Gamma(r_j)$.
	\item For each resource $r_i \in R,$ we create a set of steps $S' = \bigcup_{i=1}^k S^i$ where
	\[
	 S^i = \begin{cases}
	        \{s^i\} & \text{if $\Gamma(r_i) = \emptyset$},\\
	        \{s_j^i \mid r_j \in \Gamma(r_i)\} & \text{otherwise}.
	       \end{cases}
	\]
%
\end{itemize}
Observe that $|S'| \leq k(k-1)$.
Given a plan $\pi: S' \rightarrow U'$, we write $\pi(S^i)$ to denote $\{\pi(s) \mid s \in S^i\}$.
Let $\Pi$ be the set of all possible complete plans from $S'$ to $U'$.

We define the set of constraints $C'$ and their weights $w'_{c'}: \Pi \rightarrow \bN$ as follows.

\begin{itemize}
	\item For each $c = \BoDE{r_i, r_j} \in C$, we add constraint $c' = (s^i_j, s^j_i, =)$ to $C'$, and define
	\[
        w'_{c'}(\pi) = 
        \begin{cases}
         0 & \text{if $\pi(s^i_j) = \pi(s^j_i)$}, \\
         \ell_c & \text{otherwise}.
        \end{cases}
	\]
	Note that $c'$ is user-independent.
	
	\item For each $c = \SoDU{r_i, r_j} \in C$, we add constraint $c' = (S^i, S^j, \emptyset)$, where $c'$ is satisfied iff $\pi(S^i) \cap \pi(S^j) = \emptyset$.
	Then define $w'_{c'}(\pi)= f_{c}(|\pi(S^i) \cap \pi(S^j)|)$, where $f_c$ is the function associated with the weighted constraint $(r_i,r_j,\updownarrow,\forall)$.
	Observe that $(S^i, S^j, \emptyset)$ is a user-independent constraint.
	
	\item Let $C'$ denote the set of all constraints in $\cI'$ and define $$w'_{C'}(\pi) = \sum\limits_{c' \in C'} w'_{c'}(\pi).$$
\end{itemize}

We then define authorization weight function $\omega': U' \times 2^{S'} \rightarrow \bN$ as follows.
Initialize $\omega'(u, \emptyset) = 0$.
We set $\omega'(u, S^i) = \omega(u, \{r_i\})$.
For a subset $T \subseteq S'$, let $R_T = \{r_i \in R \mid \ T \cap S^i \neq \emptyset\}$. 
We set $\omega'(u, T) = \omega(u, R_T)$.
Given a plan $\pi: S' \rightarrow U'$, we denote $\sum_{u \in U'} \omega'(u, \pi^{-1}(u))$ by $\Omega'(\pi)$.
Finally, define the weight of $\pi$ to be $\Omega'(\pi) + w'_{C'}(\pi)$.
See Example~\ref{fig:bodE-sodU-example} for an illustration.

Based on the construction described above, we have the following lemma:

\begin{lemma}
\label{lemma:bodE-sodU-cost-optimal-preserved}
Let $\cI$ be a {\vapep}$\langle \bodE,\sodU \rangle$ instance and $\cI'$ be the {\vwsp} instance obtained from $\cI$ using the construction above.
Then $\OPT(\cI) = \OPT(\cI')$, where $\OPT(\cI)$ and $\OPT(\cI')$ denote the weights of optimal solutions of $\cI$ and $\cI'$ respectively.
Furthermore, given an optimal plan for $\cI'$, we can construct an optimal authorization relation for $\cI$ in polynomial time.
\end{lemma}

\begin{proof}
We first prove that $\OPT(\cI) \leq \OPT(\cI')$.
Let $\pi: S' \rightarrow U'$ be an optimal plan for the instance $\cI'$.
We construct $A$ for the instance $\cI$ as follows.
For all $i \in [k]$, if $u \in \pi(S^i)$, then we put $(u, r_i)$ into $A$.
This completes the construction of $A$ from $\pi$.
Since $\pi$ is complete, $A$ is also complete.
This can be implemented in polynomial time.
Observe that $r_i \in A(u)$ if and only if there exists $s \in S^i$ such that $s \in \pi^{-1}(u)$.
Equivalently, suppose that $T = \pi^{-1}(u)$.
Then, $R_T = A(u)$.
It implies that $\omega'(u, \pi^{-1}(u)) = \omega(u, A(u))$.
Hence, we have
\[
\Omega'(\pi) = \sum\limits_{u \in U'} \omega'(u, \pi^{-1}(u)) = \sum\limits_{u \in U} \omega(u, A(u)) = \Omega(A)
\]

We now prove that $w_C(A) \leq w'_{C'}(\pi)$.
Consider a $\bodE$ constraint $c = \BoDE{r_i, r_j} \in C$.
Then, $r_j \in \Gamma(r_i)$, and $r_i \in \Gamma(r_j)$, and the corresponding constraint in $\cI'$ is $c' = (s^i_j, s^j_i, =)$.
By construction if $\pi(s^i_j) = \pi(s^j_i)$, then there exists $u \in A(r_i) \cap A(r_j)$.
Hence, $w_c(A) \leq w'_{c'}(\pi)$.
Now consider an $\sodU$ constraint $c = \SoDU{r_i, r_j}$.
The corresponding constraint in $\cI'$ is $c' = (S^i, S^j, \emptyset)$.
Observe that by construction, if $\pi(S^i) \cap \pi(S^j) = \emptyset$, then $A(r_i) \cap A(r_j) = \emptyset$.
Otherwise, if $|\pi(S^i) \cap \pi(S^j)| = t > 0$, then by construction $|A(r_i) \cap A(r_j)| = t$.
Hence, $w'_{c'}(\pi) = w_c(A)$.
We obtain an authorization relation $A$ in polynomial time such that $w_{C}(A) + \Omega(A) \leq w'_{C'}(\pi) + \Omega'(\pi)= \OPT(\cI')$.
Therefore, $\OPT(\cI) \leq \OPT(\cI')$.

To complete the proof we prove that $\OPT(\cI) \geq \OPT(\cI')$.
Let $A$ be an optimal authorization relation for $\cI$.
We construct $\pi: S' \rightarrow U'$ as follows.
If $\Gamma(r_i) = \emptyset$, we choose an arbitrary $u \in A(r_i)$ and set $\pi(s^i) = u$.
Otherwise, $\Gamma(r_i) \neq \emptyset$, and two cases may arise.
\begin{itemize}
	\item For $r_j \in \Gamma(r_i)$, $\BoDE{r_i, r_j}$ is satisfied by $A$. Then, we choose an arbitrary $u \in A(r_i) \cap A(r_j)$ and set $\pi(s^i_j) = \pi(s^j_i) = u$.
	\item For $r_j \in \Gamma(r_i)$, $\BoDE{r_i, r_j}$ is not satisfied by $A$. Then, we just choose arbitrary $u \in A(r_i), v \in A(r_j)$ and set $\pi(s^i_j) = u$, and $\pi(s^j_i) = v$.
\end{itemize}
This completes the construction of $\pi$. Note that $\pi$ is complete.

Let $T = \pi^{-1}(u)$.
Observe that by construction, if $u \in \pi(S^i)$, then $u \in A(r_i)$.
Equivalently, if there exists $i \in [k]$ such that $\pi^{-1}(u) \cap S^i \neq \emptyset$, then $r_i \in A(u)$.
Therefore, $R_T \subseteq A(u)$.
Using the monotonicity property of $\omega$, we have that $\omega(u, R_T) \leq \omega(u, A(u))$.
This means that $\omega'(u, \pi^{-1}(u)) = \omega(u, R_T) \leq \omega(u, A(u))$.
Therefore, we have the following:
\[
\Omega'(\pi) = \sum\limits_{u \in U'} \omega'(u, \pi^{-1}(u)) \leq \sum\limits_{u \in U}  \omega(u, A(u)) = \Omega(A) 
\]
Hence, $\Omega'(\pi) \leq \Omega(A)$.

Consider a $\bodE$ constraint $c = \BoDE{r_i, r_j} \in C$.
By construction, $c$ is satisfied by $A$ if and only if $c' = (s^i_j, s^j_i, =)$ is satisfied by $\pi$.
Hence, $w_{c'}(\pi) = w_c(A)$.
On the other hand, consider an $\sodU$ constraint $c = \SoDU{r_i, r_j} \in C$. If $c$ is satisfied by $A$, then by construction $c' = (S^i, S^j, \emptyset)$ is also satisfied by $A$.
Finally, if $c$ is violated by $A$, then let $t = |A(r_i) \cap A(r_j)| > 0$.
By construction, $\pi(S^i) \cap \pi(S^j) \subseteq A(r_i) \cap A(r_j)$.
Hence, $w'_{c'}(\pi) \leq w_c(A)$, implying $w'_{C'}(\pi) \leq w_C(A)$.
Therefore, $\OPT(\cI') \leq \OPT(\cI)$.
%
\end{proof}

\begingroup
\begin{figure*}[t]
\centering
	\begin{subfigure}[b]{0.4\linewidth}
		\centering
		\includegraphics[scale=.25]{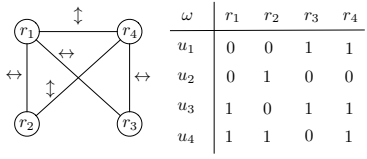}

		\caption{\vapep{} instance}
		\label{fig:sec-6-APEP}
	\end{subfigure}
	\hfill
	\begin{subfigure}[b]{0.55\linewidth}
		\centering
		\includegraphics[scale=.25]{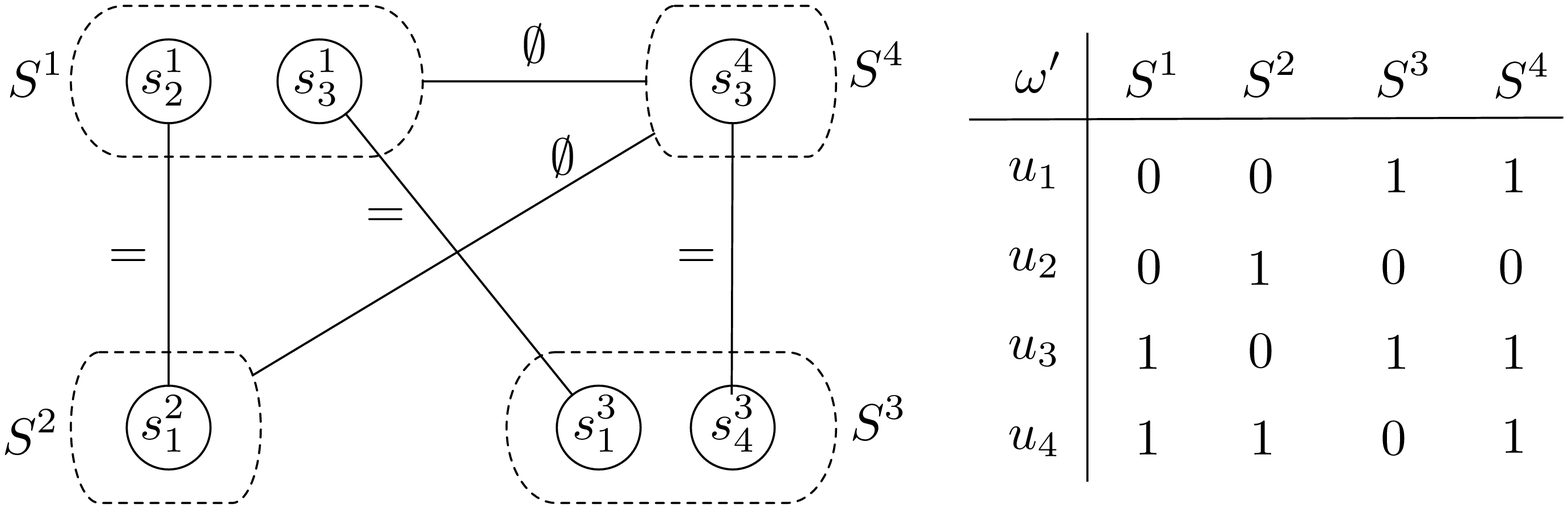}
		\caption{\vapep{} instance}
		\label{fig:sec-6-WSP}
	\end{subfigure}
	\caption{An example illustrating the construction. Note that $R = \{r_1, r_2,r_3,r_4\}$ and $S' = S^1 \cup S^2 \cup S^3 \cup S^4$.} 
\label{fig:bodE-sodU-example}
\end{figure*}
\endgroup

{\bf Example~\ref{fig:bodE-sodU-example}.}
We illustrate the construction of {\vwsp} instance from {\vapep} instance, and the proof of Lemma~\ref{lemma:bodE-sodU-cost-optimal-preserved} using  Figure~\ref{fig:bodE-sodU-example}.
In addition, given an optimal solution for the corresponding {\vwsp} instance, we illustrate how to construct an optimal solution of the {\vapep} instance as described in Lemma~\ref{lemma:bodE-sodU-cost-optimal-preserved}.
As per the figure, there are three ${\bodE}$ constraints ($c_1 = \BoDE{r_1,r_2}$, $c_2 = \BoDE{r_1, r_3}$, and $c_3=\BoDE{r_3, r_4}$) and two $\sodU$ constraints ($c_4 = \SoDU{r_1,r_4}$, and $c_5 = \SoDU{r_2, r_4}$).
We define their weights as follows.
For an authorization relation $A \subseteq U \times R$ and $i\in \{1,2,3\},$ we set $w_{c_i}(A) = 0$ if $A$ satisfies $c_i$, and $w_{c_i}(A) = 1,$ otherwise.
Let $w_{c_4}(A) = 0$ if $A$ satisfies $c_4$;
otherwise, $w_{c_4}(A) = |A(r_1) \cap A(r_4)|$. 
Similarly, $w_{c_5}(A) = 0$ if $A$ satisfies $c_5$;
otherwise, $w_{c_5}(A) = |A(r_2) \cap A(r_4)|$.

For every $\emptyset  \neq T \subseteq R$ and $u \in U,$  let $\omega(u, T) = \sum_{r \in T} \omega(u, \{r\}).$

Observe that in this example there is no authorization relation $A$ such that $w_C(A) + \Omega(A) = 0$.
It means that if we look for an authorization relation $A$ such that $w_C(A) = 0$, then we will have $\Omega(A) > 0$.
Consider an authorization relation $A^*$ such that $A^*(r_1) = \{u_1, u_2\}, A^*(r_2) = \{u_1, u_3\}, A^*(r_3) = u_2$, and $A^*(r_4) = \{u_4\}$.
Observe that $w_C(A^*) = 0$ but $\Omega(A^*) = 1$.
Conversely, if we look for an authorization relation $A$ such that $\Omega(A) = 0$, then we will have $w_C(A) > 0$.

Consider the {\vwsp} instance constructed in this example.
Based on the construction, $c_1' = (s^1_2, s^2_1, =), c_2' = (s^1_3, s^3_1, =), c_3' = (s^3_4, s^4_3, =), c_4' = (S^1, S^4, \emptyset)$, and $c_5' = (S^2, S^4, \emptyset)$.
Observe that for a given plan $\pi: S' \rightarrow U'$, we have the following:
\begin{itemize}
	\item $w_{c_1'}(\pi) = 0$ if $\pi$ satisfies $c_1'$ and $w_{c_1'}(\pi) = 1$ otherwise,
	\item $w_{c_2'}(\pi) = 0$ if $\pi$ satisfies $c_2'$ and $w_{c_2'}(\pi) = 1$ otherwise,
	\item $w_{c_3'}(\pi) = 0$ if $\pi$ satisfies $c_3'$ and $w_{c_3'}(\pi) = 1$ otherwise,
	\item $w_{c_4'}(\pi) = 0$ if $\pi$ satisfies $c_4'$ and $w_{c_4'}(\pi) = |\pi(S^1) \cap \pi(S^4)|$ otherwise, and
	\item $w_{c_5'}(\pi) = 0$ if $\pi$ satisfies $c_5'$ and $w_{c_5'}(\pi) = |\pi(S^2) \cap \pi(S^4)|$ otherwise.
\end{itemize}

Consider an optimal plan $\pi: S' \rightarrow U'$ defined as follows:
	$\pi(s^1_2) = u_1$,
	 $\pi(s^1_3) = u_2$,
	 $\pi(s^2_1) = u_1$,
	 $\pi(s^3_1) = u_2$,
	 $\pi(s^3_4) = u_4$, and
	 $\pi(s^4_3) = u_4$.
%
Observe that $w'_{C'}(\pi) = 0$ as all constraints $c_1',c_2',c_3',c_4'$, and $c_5'$ are satisfied. 
Then, $\Omega'(u_1,\{s^1_2, s^2_1\}) = \Omega(u_1, \{r_1,r_2\}) = 0$, $\Omega'(u_2,\{s^1_3,s^3_1\}) = \Omega(u_2,\{r_1,r_3\}) = 0$, and $\Omega'(u_4,\{s^3_4,s^4_3\}) = \Omega(u_4,\{r_3, r_4\}) = 1$.
Finally, $\Omega'(u_3,\emptyset) = 0$.
Thus, $\Omega'(\pi) = 1$.

We construct $A$ from $\pi$ as in the first part of the the proof of Lemma \ref{lemma:bodE-sodU-cost-optimal-preserved}:
	$A(r_1) = \{u_1, u_2\}$,
	$A(r_2) = \{u_1\}$,
	$A(r_3) = \{u_2, u_4\}$, and
	$A(r_4) = \{u_4\}.$
Observe that $w_{C}(A) = 0$ as all constraints $c_1,\ldots,c_5$ are satisfied by $A$ and $\Omega(A) = 1$.

We can now state the main result of this section.

\begin{theorem}
\label{thm:bodE-sodU-FPT}
{\vapep}$\ev{\bodE, \sodU}$ is fixed-parameter tractable and can be solved in $\cO^*(4^{k^2 \log k})$ time.
\end{theorem}

\begin{proof}
Let $\cI = (R, U, C, \omega)$ be an instance of {\vapep}$\ev{\bodE,\sodU}$.
We construct an instance $\cI' = (S', U, C', \omega')$ of {\vwsp} in polynomial time.
We then invoke Theorem~\ref{thm1} to obtain an optimal plan $\pi: S' \rightarrow U'$.
Finally, we invoke Lemma~\ref{lemma:bodE-sodU-cost-optimal-preserved} to construct an optimal authorization relation $A$ for $\cI$ such that  $\Omega(A) + w_{C}(A) = \OPT(\cI')$.
The algorithm described in Theorem~\ref{thm1} runs in $\cO^*(2^{|S'|\log |S'|})$ time.
Since $|S'| \leq k(k-1)$, the running time of this algorithm to solve {\vapep}$\ev{\bodE, \sodU}$ is $\cO^*(4^{k^2 \log k})$.
\end{proof}

\section{Using Valued APEP to address resiliency in workflows}
\label{sec:resilient-wsp}

Resiliency, in the context of access control, is a generic term for the ability of an organization to continue to conduct business operations even when some authorized users are unavailable~\cite{LiWaTr09}. Resiliency is particularly interesting when an organization specifies authorization policies and separation of duty constraints, as is common in workflow systems, as separation of duty constraints become harder to satisfy when fewer (authorized) users are available.

Early work by Wang and Li showed that determining whether a workflow specification is resilient is a hard problem~\cite{LiWaTr09}. More recent work has established the precise complexity of determining static resiliency~\cite{Fong19}, and that the problem is FPT, provided all constraints in the workflow specification are user-independent~\cite{CramptonGKW17}. 

We introduce the idea of an extended plan for a workflow specification and define resiliency in the context of an extended plan. 
We then explain how Valued APEP can be used to compute extended plans of minimal cost. In Section~\ref{sec:formulations}, we provide two MIP formulations for Valued APEP. Then in Section~\ref{sec:experiments} we discuss our experimental framework and results making use of these two formulations.

	Suppose we are given a workflow specification (defined by a set of workflow steps $S$, a set of users $U$, an authorization relation $\hat{A} \subseteq U \times S$ and a set of constraints $C$) and an integer $\tau \ge 0$.
	We call a function $\Pi : S \rightarrow 2^U$ an \emph{extended plan}; we say $\Pi$ is \emph{valid} if there exists a valid plan $\pi : S \rightarrow U$ such that $\pi(s) \in \Pi(s)$ for all $s \in S$ and $\pi$ is valid. 
	We say $\Pi$ is \emph{$\tau$-resilient} if for any subset of $\tau$ users $T \subseteq U$, there exists a valid plan $\pi' : S \rightarrow (U \setminus T)$ such that $\pi'(s) \in \Pi(s)$ for all $s \in S$.
	Wang and Li introduced an alternative notion of resiliency in workflows~\cite{WaLi10}, where a workflow specification is said to be \emph{statically $t$-resilient} if for all $U' \subseteq U$ such that $|U| - |U'| \le t$, $(S,U',C,\hat{A}')$, where $\hat{A}' = \hat{A} \cap (S \times U')$, is satisfiable.
	
	The two notions of resiliency are rather different.
	Our notion requires an extended plan to be resilient, so that having committed to an extended plan for a workflow we know the instance can complete even if $\tau$ users are unavailable.
	In contrast, Wang and Li require that the workflow specification itself is resilient.
	Crampton, Gutin, Karapetyan and Watrigant showed that determining whether a workflow is statically $t$-resilient is FPT~\cite{CramptonGKW17} for WSP with UI constraints only.
	
	It is not obvious that the methods used by Crampton \emph{et al.}\ can be adapted to determine whether there exists a $\tau$-resilient extended plan.
	Nor is it obvious whether the problem of deciding if there exists a $\tau$-resilient extended plan can be framed as an instance of APEP\@.
	
	APEP, however, can be used to produce an extended plan that is $\tau$-resilient.
	Moreover, \vapep{} can be used to solve the softer problem of finding an extended plan that aims to be $\tau$-resilient (but may not be) and that also minimises the number of users involved.

To generate a $\tau$-resilient extended plan for a WSP instance $(S, U, C, \hat{A})$ with SoD constraints, we can produce the following APEP instance $(R', U', C', \hat{A}')$:
\begin{itemize}
	\item
	Let $R' = S$, $U' = U$ and $\hat{A}' = \hat{A}$.
	
	\item
	Let $C' = \emptyset$.
	For every $c \in C$, add a corresponding SoD$_U$ to $C'$ (recall that we consider WSP with SoD constraints only).
	
	\item
	Add Cardinality-Lower-Bound constraints $(r, \ge, \tau + 1)$ for every $r \in R$.
\end{itemize}
Any authorization relation $A$ that satisfies such an APEP instance is a $\tau$-resilient extended plan in the original WSP instance.
(Note that it is sufficient for an extended plan $\Pi$ to be a solution of an APEP instance, however it is not necessary; some $\tau$-resilient extended plans may not be solutions for an APEP instance.)

The requirement to have at least $\tau + 1$ users assigned to each resource may lead to solutions that involve too many users.
In practice, we may want to keep the number of users involved in $\Pi$ as small as possible.
Also, where an instance is not $\tau$-resilient, we may want to accept solutions that are not completely $\tau$-resilient, i.e.\ solutions where excluding $\tau$ users may render the extended plan invalid to some (acceptably limited) extent.

To meet the above requirements, we can use the \vapep{} to model $\tau$-resiliency in WSP.
Let $p_\text{SoD}$ and $p_\text{Card}$ be penalties for violation of the corresponding constraints.
Let $p_A$ be a penalty for assigning a user to a resource to which this user is not authorized.
Compose a APEP instance $(R, U, C, \hat{A})$ as described above and replace each constraint with the following weighted constraints:
\begin{itemize}	
	\item
	For every SoD$_U$ constraint $c = (r_1, r_2, \updownarrow, \forall)$, let $w_c(A) = p_\text{SoD} \cdot |A(r_1) \cap A(r_2)|$; i.e., there is a $p_\text{SoD}$ penalty for every user assigned to both resources in the scope.
	
	\item
	For every cardinality-lower-bound constraint $c = (r, \ge, \tau + 1)$, let $w_c = \max \{ 0,\, p_\text{Card} \cdot (\tau + 1 - |A(r)|) \}$. 
	
	\item
	Finally, we add a constraint $c$, which we call \emph{User Count Constraint}, with the scope $R$ such that $w_c = f_\Pi(|A(R)|)$, where $f_\Pi(\cdot)$ is a monotonically growing function.
	Specifically, we use $w_c = |A(R)|^2$.
\end{itemize}
Also let $\omega(u, T) = p_A \cdot \ell$, where $\ell$ is the number of resources $r \in T$ such that $(r, u) \notin \hat{A}$.
In other words, there is a $p_A$ penalty for each unauthorized assignment of a user to a resource.

Assuming $p_A$, $p_\text{SoD}$ and $p_\text{Card}$ are positive numbers and the original APEP instance is satisfiable, a solution to this \vapep{} instance will be a $\tau$-resilient extended plan with at least $\tau+1$ users assigned to each resource, with all the SoD$_U$ constraints and authorizations satisfied and with the number of users involved in the extended plan minimized.

\section{Mixed Integer Formulations of \vapep{}}
\label{sec:formulations}

While it is common to implement bespoke algorithms to exploit the FPT properties of a problem, it was noted recently that off-the-shelf solvers may also be efficient on such problems given appropriate formulations~\cite{KarapetyanPGG19,KaGu2021}.
In this section we give two mixed integer programming (MIP) formulations of the \vapep{}\@.
The formulation given in Section~\ref{sec:naive} is a straightforward interpretation of the problem; it uses binary variables to define an assignment of users to resources.
The formulation given in Section~\ref{sec:up}, however, makes use of the concept of user profiles, used earlier to prove FPT results.
While both formulations are generic enough to support any \vapep{} constraints, we focus on the constraints used to model $\tau$-resiliency of WSP extended plans, see Section~\ref{sec:resilient-wsp}.

\subsection{Naive formulation}
\label{sec:naive}

The \emph{Naive} formulation of a \vapep{} instance $(R, U, C, \omega)$ is based on binary variables $x_{r,u}$ linking resources to users; $x_{r,u} = 1$ if and only if user $u$ is assigned to resource $r$.

The core of the formulation is as follows:
\begin{align}
& \text{minimize } \sum_{c \in C} p_c + p_A \cdot \sum_{(r, u) \notin \hat{A}} x_{r,u} \\
& \text{subject to} \nonumber \\
& x_{r,u} \in \{ 0, 1 \} 
	&& \forall r \in R,\ \forall u \in U,\\
& p_c \in [0, \infty]
	&& \forall c \in C.
\end{align}
The encodings of the \vapep{} constraints linking the solution to variables $p_c$ are discussed below.

The User Count constraint $c$ is encoded as follows:
\begin{align}
& p_c = f_\Pi(z), \\
& z = \sum_{u \in U} y_u, \\
& y_u \ge x_{r,u} 
	&& \forall r \in R,\ \forall u \in U,\\
& y_u \in [0, 1]
	&& \forall u \in U,\\
& z \in [0, n].
\end{align}
Variable $z$ is introduced to count the number $|A(R)|$ of users generated by the solution.
The formulation depends on the function $f_\Pi$; for $f_\Pi(z) = z^2$, we use the following encoding as a discretization of a parabola:
\begin{align}
& p_c \ge f_i(z)
	&& \forall i \in \{1, 2, \ldots, n - 1 \},
\label{eq:parabola}
\end{align}
{where $f_i(z) = (2i + 1) z - (i + 1) i$.
An illustration of how it enforces $p_c \ge z^2$ is given in Figure~\ref{fig:parabola}.
Note that $z^2 \ge f_i(z)$ for every integer $z$ and $i$.}

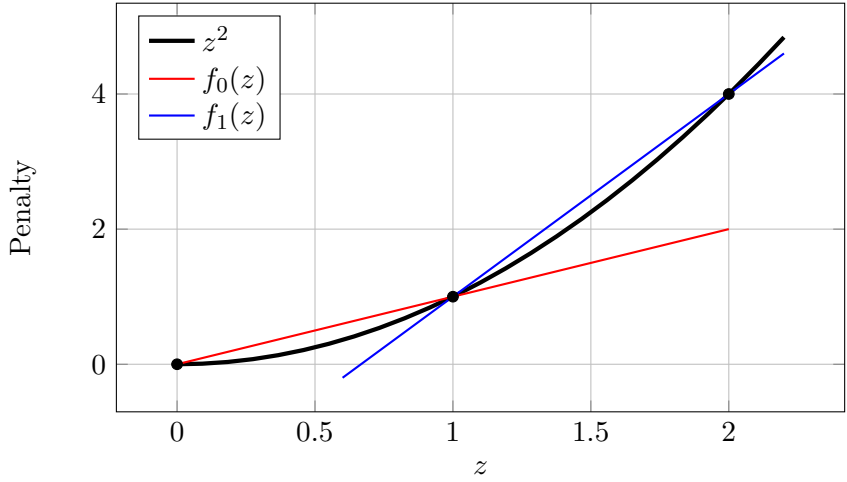
\begin{figure}[htb]
\begin{tikzpicture}
\begin{axis}[
    xlabel={$z$},
    ylabel={Penalty},
    legend pos={north west},
    grid=major,
    height=7cm,
    width=0.7\linewidth,
    legend cell align={left},
]
    \addplot[black, ultra thick, domain=0:2.2] {x^2};
    \addlegendentry{$z^2$}
    
    \addplot[red, thick, domain=0:2] {x};
    \addlegendentry{$f_0(z)$}

    \addplot[blue, thick, domain=0.6:2.2] {3 * x - 2};
    \addlegendentry{$f_1(z)$}


    \addplot[black, only marks, domain=0:2, samples=3] {x^2};
\end{axis}
\end{tikzpicture}

\caption{Illustration of how equations (\ref{eq:parabola}) enforce $p_c \ge z^2$.}
\label{fig:parabola}
\end{figure}

Each Cardinality-Lower-Bound constraint $c = (r, \ge, \tau+1)$ is encoded as follows:
\begin{align}
& p_c \ge p_\text{Card} \cdot \left[ (\tau + 1) - \sum_{u \in U} x_{r,u}\right].
\end{align}


Each SoD constraint $c = (r_1, r_2, \updownarrow, \forall)$ is encoded as follows:
\begin{align}
& p_c = p_\text{SoD} \cdot \sum_{u \in U} y_u, \\
& y_u >= x_{r_1,u} + x_{r_2,u} - 1
	&& \forall u \in U,\\
& y_u \in \{0, 1\}
	&& \forall u \in U.
\end{align}

\subsection{User-Profile formulation}
\label{sec:up}

The \emph{User-Profile} (UP) formulation is based on the concept of user profiles making it an FPT-aware formulation.
Let $\mathcal{T}$ be the power set of $R$.
The UP formulation defines a binary variable $x_{T, u}$ for every $T \in \mathcal{T}$ and $u \in U$.
We then require that each user $u$ is assigned exactly one $T \in \mathcal{T}$.

The core of the formulation is as follows:
\begin{align}
& \omit\rlap{minimize $\displaystyle \sum_{c \in C} p_c + p_A \cdot \sum_{T \in \mathcal{T}} \sum_{u \in U} x_{T,u} \cdot | \{ r \in T : (r, u) \notin \hat{A} \} |$} \\
& \text{subject to} \nonumber \\
& x_{T,u} \in \{ 0, 1 \} 
	&& \forall T \in \mathcal{T},\ \forall u \in U, \hspace{8em}\\
& p_c \in [0, \infty]
	&& \forall c \in C.
\end{align}
The encodings of the \vapep{} constraints linking the solution to variables $p_c$ are discussed below.

The encoding of the User Count constraint $c$ is very similar to its implementation in the Naive formulation:
\begin{align}
& p_c = f_\Pi(z), \\
& z = \sum_{u \in U} y_u, \\
& y_u \ge x_{T,u} 
	&& \forall T \in \mathcal{T} \setminus \{ \emptyset \},\ \forall u \in U,\\
& y_u \in [0, 1]
	&& \forall u \in U,\\
& z \in [0, n].
\end{align}
{Specifically, we use $f_\Pi(z) = z^2$, which we encode as follows (see (\ref{eq:parabola}) for details):}
\begin{align}
& p_c \ge f_i(z)
	&& \forall i \in \{1, 2, \ldots, n - 1 \}.
\end{align}

Each Cardinality-Lower-Bound constraint $c = (r, \ge, \tau+1)$ is encoded as follows:
\begin{align}
& p_c \ge p_\text{Card} \cdot \left[ (\tau + 1) - \sum_{T \in \mathcal{T},\, r \in T}\sum_{u \in U} x_{T,u}\right].
\end{align}

Each SoD constraint $c = (r_1, r_2, \updownarrow, \forall)$ is encoded as follows:
\begin{align}
& p_c = p_\text{SoD} \cdot \sum_{T \in \mathcal{T},\, r_1, r_2 \in T} \sum_{u \in U} x_{T,u}.
\end{align}

\section{Computational experiments}
\label{sec:experiments}

The aims of our computational study are to:
\begin{enumerate}
	\item
	design an instance generator, to support future experimental studies of the \vapep{} and enable fair comparison of \vapep{} solution methods;
	
	\item
	give a new approach to address resiliency in WSP;
	
	\item
	compare the performance of the two formulations discussed in Section~\ref{sec:formulations};
	
	\item
	test if either of the formulations has FPT-like running times, i.e.\ scales polynomially with the instance size given that the small parameter is fixed;
	
	\item
	analyse the structure of optimal solutions and how it depends on the instance generator inputs; and	
	
	\item
	make the instance generator and the solvers based on the two formulations publicly available.
\end{enumerate}


\subsection{Benchmark instances}
\label{sec:instance-generator}

For our computational experiments, we built a pseudo-random instance generator for the \vapep{}\@.
It is designed around the concept of $\tau$-resiliency of WSP, see Section~\ref{sec:resilient-wsp}.
The inputs of the instance generator are detailed in Table~\ref{tab:inputs}.

\begin{table}[htb]
\renewcommand{\arraystretch}{1.3}
\begin{tabular}{lp{30em}l}
\toprule
Input & Description & Default value \\
\midrule
$n$ & the number of users, also referred to as the size of the problem & -- \\
$k$ & the number of WSP steps & $\lfloor 0.1n \rfloor$ \\
$\tau$ & the desired degree of $\tau$-resiliency, i.e.\ the number of users that can be excluded from the extended plan & $\lfloor 0.05 n \rfloor$ \\
$\alpha$ & an input enabling us to adjust the balance between the penalties associated with workflow and resiliency violations & 1 \\
\bottomrule
\end{tabular}


\caption{The instance generator inputs and their default values.}
\label{tab:inputs}
\end{table}

%
%
%

While inputs $n$ and $k$ define the size of the instance and $\tau$ is an inherent input of $\tau$-resiliency, $\alpha$ is an artifact of the experimental set-up, introduced to control the weight of the workflow constraints and authorizations relative to resiliency.  
Varying the value of $\alpha$ enables us to investigate the effects of emphasizing the importance of satisfying workflow constraints (over resiliency) and vice versa.  
The greater the value of $\alpha$, the greater the penalties for violating workflow constraints and authorizations, meaning that satisfying resiliency becomes correspondingly less significant.

The instance generator first creates a WSP instance $(S, U, C, \hat{A})$ and then converts that instance into a \vapep{} instance $(R, U, C', \omega)$, as described in Section~\ref{sec:resilient-wsp}.
The constraint penalties are set as following: $p_\text{SoD} = 10\alpha$, $p_\text{Card} = 10$ and $p_A = \alpha$.
The WSP instance is generated in the following way:
\begin{enumerate}
	\item
	create steps $S = \{ s_1, s_2, \ldots, s_k \}$ and users $U = \{ u_1, u_2, \ldots, u_n \}$;
		
	\item
	the authorizations are created in the same way as in the WSP instance generator, see~\cite{KarapetyanPGG19}: for each user $u \in U$, select randomly and uniformly from $[1, \lfloor 0.5 \cdot (k - 1) \rfloor]$ the number of steps to which $u$ is authorized, and then randomly select which steps they are authorized to; and
		
	\item
	produces $q_\text{SoD}$ constraints SoD, selecting the scope of each of them randomly and independently (the generator may produce several SoD constraints with the same scope).
\end{enumerate}


%
%
%
%
%
%

\subsection{$t$-wboundness of the User Count constraint}

It is a trivial observation that the User Count constraint is 0-wbounded.
However, we can also establish the $t$-wboundness of a set of constraints $C$, where $C$ includes a User Count constraint.

\begin{proposition}
\label{rem:user-count}
Let $C_\text{SoD}$ be a set of SoD$_U$ constraints.
Let $C_\text{Card}$ be a set of Cardinality-Lower-Bound constraints with the penalty function $w_c(A) = \max \{ 0,\, p_\text{Card} \cdot (\ell - |A(r)|) \}$ for some $\ell$.
Let $c_\Pi$ be a User Count constraint with the penalty function $w_{c_\Pi}(A) = |A(R)|^2$.
Let $C = C_\text{SoD} \cup C_\text{Card} \cup \{ c_\Pi \}$.
Then $C$ is $0.5k \cdot (|C_\text{Card}| \cdot p_\text{Card} + 1)$-bounded.
\end{proposition}
\begin{proof}
Let us assume that $A$ is an authorization relation that minimizes $w_C(A)$ and that $|A| > 0.5k \cdot (|C_\text{Card}| \cdot p_\text{Card} + 1)$.
Observe that $|A(R)| > 0.5 \cdot (|C_\text{Card}| \cdot p_\text{Card} + 1)$.
We will show a contradiction by constructing an authorization relation $A'$ such that $w_C(A') < w_C(A)$ and $|A'(R)| = |A(R)| - 1$.

Select an arbitrary user $u \in U$ such that $A(u) \neq \emptyset$.
Let $A'$ be an authorization relation such that $A'(u') = A(u')$ for $u' \in U \setminus \{ u \}$ and $A'(u) = \emptyset$.
(Effectively, we exclude one user involved in the authorization relation.)
Note that
\begin{enumerate}
	\item
	$w_c(A') - w_c(A) \le 0$ for every $c \in C_\text{SoD}$ as excluding a user cannot increase the SoD penalty.
	
	\item
	$w_c(A') - w_c(A) \le p_\text{Card}$ as the penalty of a Cardinality-Lower-Bound constraint can only increase by $p_\text{Card}$ following an exclusion of a single user.
	
	\item
	 $w_{c_\Pi}(A') - w_{c_\Pi}(A) = |A'(R)|^2 - |A(R)|^2 = (|A(R)| - 1)^2 - |A(R)|^2 = -2|A(R)| + 1$.
\end{enumerate}
Hence, $w_C(A') - w_C(A) \le |C_\text{Card}| \cdot p_\text{Card} - 2|A(R)| + 1$.
Since $|A(R)| > 0.5 \cdot (|C_\text{Card}| \cdot p_\text{Card} + 1)$,
\[
w_C(A') - w_C(A) < |C_\text{Card}| \cdot p_\text{Card} - 2 \cdot 0.5 \cdot (|C_\text{Card}| \cdot p_\text{Card} + 1) + 1 = 0.
\]
In other words, $w_C(A') < w_C(A)$ which is a contradiction to the assumption that $A$ minimizes $w_C(A)$.
Hence, an optimal authorization relation $A$ cannot be of size greater than $0.5k \cdot (|C_\text{Card}| \cdot p_\text{Card} - 1)$, i.e.\ $C$ is $0.5k \cdot (|C_\text{Card}| \cdot p_\text{Card} + 1)$-wbounded.
\end{proof}

Proposition~\ref{rem:user-count} is important because it shows that the instances produced by our instance generator are $t$-wbounded and that $t$ does not depend on $\alpha$, $\tau$ or $n$.
Hence, an FPT algorithm is expected to scale polynomially with $n$ if $k$ is fixed, even though $\tau$ is a function of $n$.
We will use this as a test for FPT-like running times.

\subsection{Computational results}

We used IBM CPLEX~20.1 to solve the MIP formulations.
The formulations were generated using Python~3.8.8 scripts available at~\url{doi.org/10.17639/nott.7124}.
The experiments were conducted on a Dell XPS~15~9570 with Intel i7-8750H CPU (2.20~GHz) and 32~GB of RAM\@.
CPLEX was allowed to use all the CPU cores.
Only one instance of CPLEX would run at any point in time.
Each experiment was repeated 10 times for 10 different instances produced with different random number generator seed values.
The results reported in this section are the averages over the 10 runs.

\newenvironment{customlegend}[1][]{%
        \begingroup
        \csname pgfplots@init@cleared@structures\endcsname
        \pgfplotsset{#1}%
    }{%
        \csname pgfplots@createlegend\endcsname
        \endgroup
    }%
    \def\addlegendimage{\csname pgfplots@addlegendimage\endcsname}

\pgfplotsset{
	my subfig/.style={
		compat=newest,
		width=1\textwidth,
		grid=major,
		typeset ticklabels with strut, 
		inner ysep=1pt},
	naive solver/.style={blue, ultra thick},
	up solver/.style={red, ultra thick},
	total penalty/.style={black, thick},
	sod penalty/.style={olive, thick},
	auth penalty/.style={blue, thick},
	user count penalty/.style={green, thick},
	card penalty/.style={red, thick},
}

\begin{figure}[t]
\centering
	\begin{subfigure}[b]{1\linewidth}
		\begin{center}
		\begin{tikzpicture}
			\begin{customlegend}[
				legend columns=-1,
				legend style={align=left,/tikz/every even column/.append style={column sep=1em}},
        		legend entries={Naive, UP}]
	        \addlegendimage{naive solver, line legend}
	        \addlegendimage{up solver, line legend}
        	\end{customlegend}
		\end{tikzpicture}
		\end{center}
	\end{subfigure}


	\begin{subfigure}[t]{0.49\linewidth}
		\begin{center}
		\begin{tikzpicture}[trim axis right,trim axis left]
			\begin{semilogyaxis}[
				height=4cm,
				xlabel={$n$},
				ylabel={Running time, sec},
				ymin=0,
				my subfig
			]
	
	        \addplot[naive solver] table[
				x=n,
				y=time,
				restrict expr to domain={\thisrow{alpha}}{1:1},
			] {times-k-naive.txt};
	
	        \addplot[up solver] table[
				x=n,
				y=time,
				restrict expr to domain={\thisrow{alpha}}{1:1},
			] {times-k-up.txt};
			\end{semilogyaxis}
		\end{tikzpicture}		
		\end{center}
	\end{subfigure}
	\hfill
	\begin{subfigure}[t]{0.49\linewidth}
		\begin{center}
		\begin{tikzpicture}[trim axis right,trim axis left]
			\begin{loglogaxis}[
				height=4cm,
				xlabel={$n$},
				ylabel={Running time, sec},
				ymin=0,
				my subfig
			]
	
	        \addplot[naive solver] table[
				x=n,
				y=time,
				restrict expr to domain={\thisrow{alpha}}{1:1},
			] {times-n-naive.txt};
	
	        \addplot[up solver] table[
				x=n,
				y=time,
				restrict expr to domain={\thisrow{alpha}}{1:1},
			] {times-n-up.txt};
			\end{loglogaxis}
		\end{tikzpicture}	
		\end{center}
	\end{subfigure}
	
	\bigskip
	
	\begin{subfigure}[t]{0.49\linewidth}
		\begin{center}
		\begin{tikzpicture}[trim axis right,trim axis left]
			\begin{axis}[
				height=3cm,
				xlabel={$n$},
				ylabel={$|A(R)|$},
				ymin=0,
				my subfig
			]
		
	        \addplot[black, ultra thick] table[
				x=n,
				y=users,
				restrict expr to domain={\thisrow{alpha}}{1:1},
			] {times-k-up.txt};
			\end{axis}
		\end{tikzpicture}		
		\end{center}
	\end{subfigure}
	\hfill
	\begin{subfigure}[t]{0.49\linewidth}
		\begin{center}
		\begin{tikzpicture}[trim axis right,trim axis left]
			\begin{semilogxaxis}[
				height=3cm,
				xlabel={$n$},
				ylabel={$|A(R)|$},
				ymin=0,
				my subfig
			]
	
	        \addplot[black, ultra thick] table[
				x=n,
				y=users,
				restrict expr to domain={\thisrow{alpha}}{1:1},
			] {times-n-up.txt};
			\end{semilogxaxis}
		\end{tikzpicture}		
		\end{center}
	\end{subfigure}
	
	\bigskip
		
	\begin{subfigure}[b]{1\linewidth}
		\begin{center}
		\begin{tikzpicture}
			\begin{customlegend}[
				legend columns=-1,
				legend style={align=left,/tikz/every even column/.append style={column sep=1em}},
        		legend entries={Total, Cardinality, User Count, Authorizations, SoD}]
	        \addlegendimage{total penalty, line legend}
	        \addlegendimage{card penalty, line legend}
	        \addlegendimage{user count penalty, line legend}
	        \addlegendimage{auth penalty, line legend}
	        \addlegendimage{sod penalty, line legend}
        	\end{customlegend}
		\end{tikzpicture}
		\end{center}
	\end{subfigure}
	
	\smallskip
	
	\begin{subfigure}[b]{0.49\linewidth}
		\begin{center}
		\begin{tikzpicture}[trim axis right,trim axis left]
			\begin{axis}[
				height=4cm,
				xlabel={$n$},
				ylabel={Penalty},
				ymin=0,
				my subfig
			]
			\addplot[total penalty] table[
				x=n,
				y=objective,
				restrict expr to domain={\thisrow{alpha}}{1:1},
			] {times-k-up.txt};	
	
	        \addplot[card penalty] table[
				x=n,
				y=card-penalty,
				restrict expr to domain={\thisrow{alpha}}{1:1},
			] {times-k-up.txt};
	
	        \addplot[user count penalty] table[
				x=n,
				y=users-penalty,
				restrict expr to domain={\thisrow{alpha}}{1:1},
			] {times-k-up.txt};
	
	        \addplot[auth penalty] table[
				x=n,
				y=authorisations-penalty,
				restrict expr to domain={\thisrow{alpha}}{1:1},
			] {times-k-up.txt};
			
	        \addplot[sod penalty] table[
				x=n,
				y=sod-penalty,
				restrict expr to domain={\thisrow{alpha}}{1:1},
			] {times-k-up.txt};
			\end{axis}
			\end{tikzpicture}
		\end{center}

		\caption{$k = \lfloor 0.1 n \rfloor$.}
		\label{fig:scaling-k}
	\end{subfigure}
	\hfill
	\begin{subfigure}[b]{0.49\linewidth}
		\begin{center}
		\begin{tikzpicture}[trim axis right,trim axis left]
			\begin{loglogaxis}[
				height=4cm,
				xlabel={$n$},
				ylabel={Penalty},
				ymin=0,
				my subfig
			]
			\addplot[total penalty] table[
				x=n,
				y=objective,
				restrict expr to domain={\thisrow{alpha}}{1:1},
			] {times-n-up.txt};
	
	        \addplot[card penalty] table[
				x=n,
				y=card-penalty,
				restrict expr to domain={\thisrow{alpha}}{1:1},
			] {times-n-up.txt};
	
	        \addplot[user count penalty] table[
				x=n,
				y=users-penalty,
				restrict expr to domain={\thisrow{alpha}}{1:1},
			] {times-n-up.txt};
	
	        \addplot[auth penalty] table[
				x=n,
				y=authorisations-penalty,
				restrict expr to domain={\thisrow{alpha}}{1:1},
			] {times-n-up.txt};
			
	        \addplot[sod penalty] table[
				x=n,
				y=sod-penalty,
				restrict expr to domain={\thisrow{alpha}}{1:1},
			] {times-n-up.txt};
			\end{loglogaxis}
		\end{tikzpicture}
		\end{center}

		\caption{$k = 10$.}
		\label{fig:scaling-n}
	\end{subfigure}
	\caption{Scaling of the solution time, number $|A(R)|$ of users involved in the solution and the penalties as the instance size changes.
	In all instances, $\alpha = 1$ and $\tau = \lfloor 0.05 n \rfloor$.}
\label{fig:scaling}
\end{figure}
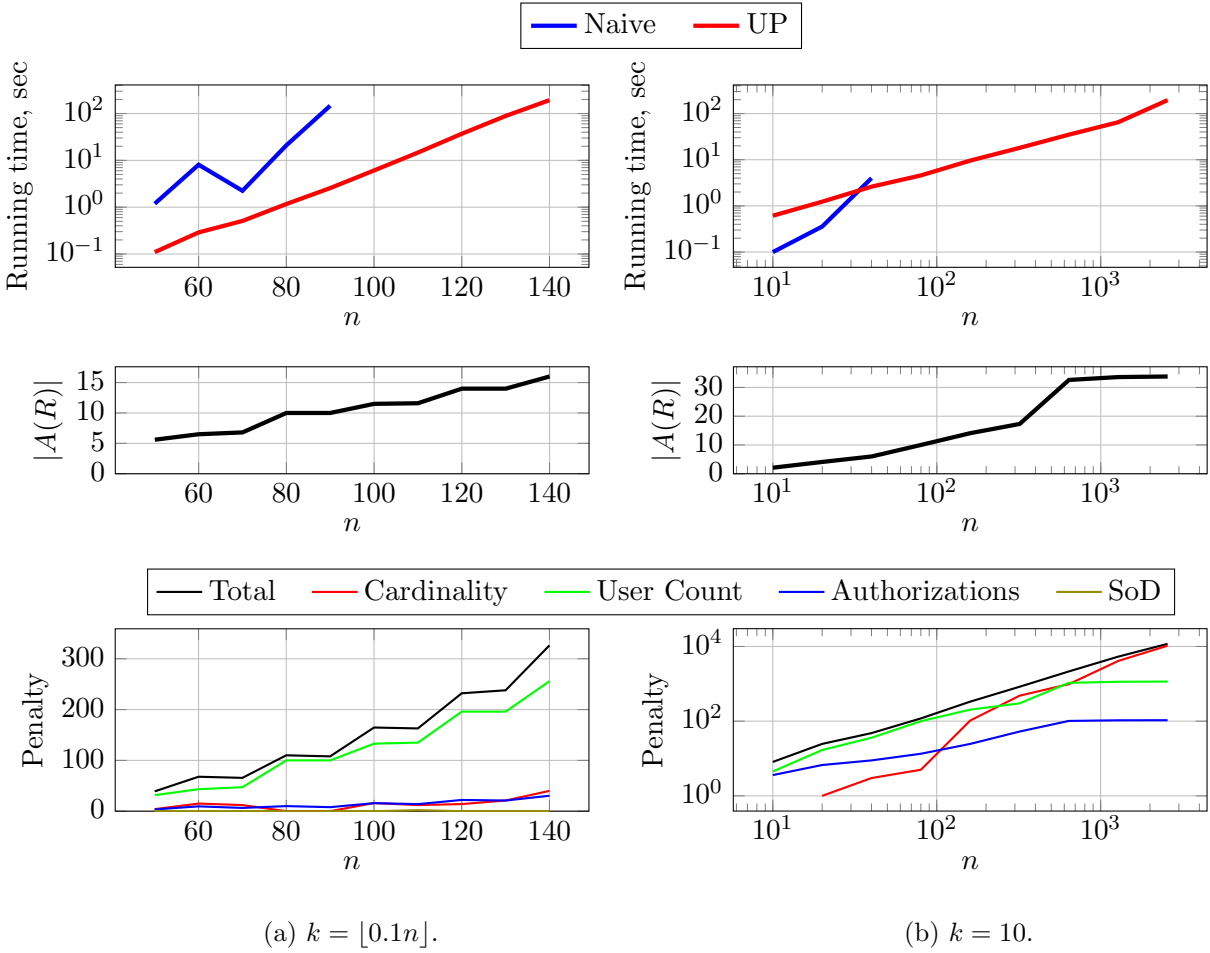

\subsubsection{Scaling}

In our first set of experiments we adjust the instance size $n$ and analyse how this affects the solution time and the optimal solution properties.
This is particularly important to understand the limitations of the methods in terms of the instance size that they can handle, as well as study the structure of solutions to large instances.
In Figure~\ref{fig:scaling-k}, we change both $k$ and $n$ (and all the associated instance generator inputs), to test how the runtime of the solvers scale.
However, as the problem is FPT, we also tested in Figure~\ref{fig:scaling-n} how the solution time and the optimal solution properties change if the value of the small parameter $k$ is fixed while the problem size $n$ changes.

We notice that the UP solver generally outperforms the Naive solver by a large margin; in fact, it scales much better, hence the gap between the solvers increases with the instance size.
The running time of the UP solver seems to be exponential only in $k$ and linear in $n$; i.e., it has FPT-like running time.
It is hard to determine how the Naive solver's running time scales as we could only obtain a few data points but it appears that its running time scales super-polynomially  even if $k$ is fixed meaning that its running time is not FPT-like.
In other words, we believe that the UP solver efficiently exploits the FPT structure of the problem whereas the Naive solver fails to do so.

When we scale both $k$ and $n$ (Figure~\ref{fig:scaling-k}), the number of users $|A(R)|$ in the optimal solutions grows linearly.
However when we fix $k$ (Figure~\ref{fig:scaling-n}), there seems to be an upper bound on $|A(R)|$.
This is consistent with our expectations; according to Proposition~\ref{rem:user-count}, the number of users is expected to be bounded by $0.5 (|C_\text{Card}| \cdot p_\text{Card} + 1) = 0.5 (10k + 1)$.
For $k = 10$, this gives us an upper bound of around 50.
The discrepancy with the practice is due to the influence of the SoD constraints and authorizations, both generating pressure to keep the number of users small.

As long as $k$ is comparable to $n$, the User Count constraint is the main cause of the penalty.
However, as $n$ gets bigger relative to $k$, the cardinality constraints penalty begins to dominate.
This is due to the relation between $n$ and $\tau$; while the number of users stays unchanged as we increase $n$, the value of $\tau$ grows as does the penalty caused by violations of the cardinality constraints.
With a few minor exceptions, all the SoD constraints are satisfied in all the experiments, whereas authorizations are often violated to a small extent.
This imbalance is due to the 10-fold difference in the corresponding penalties.

\subsubsection{Sensitivity to $\alpha$ and $\tau$}

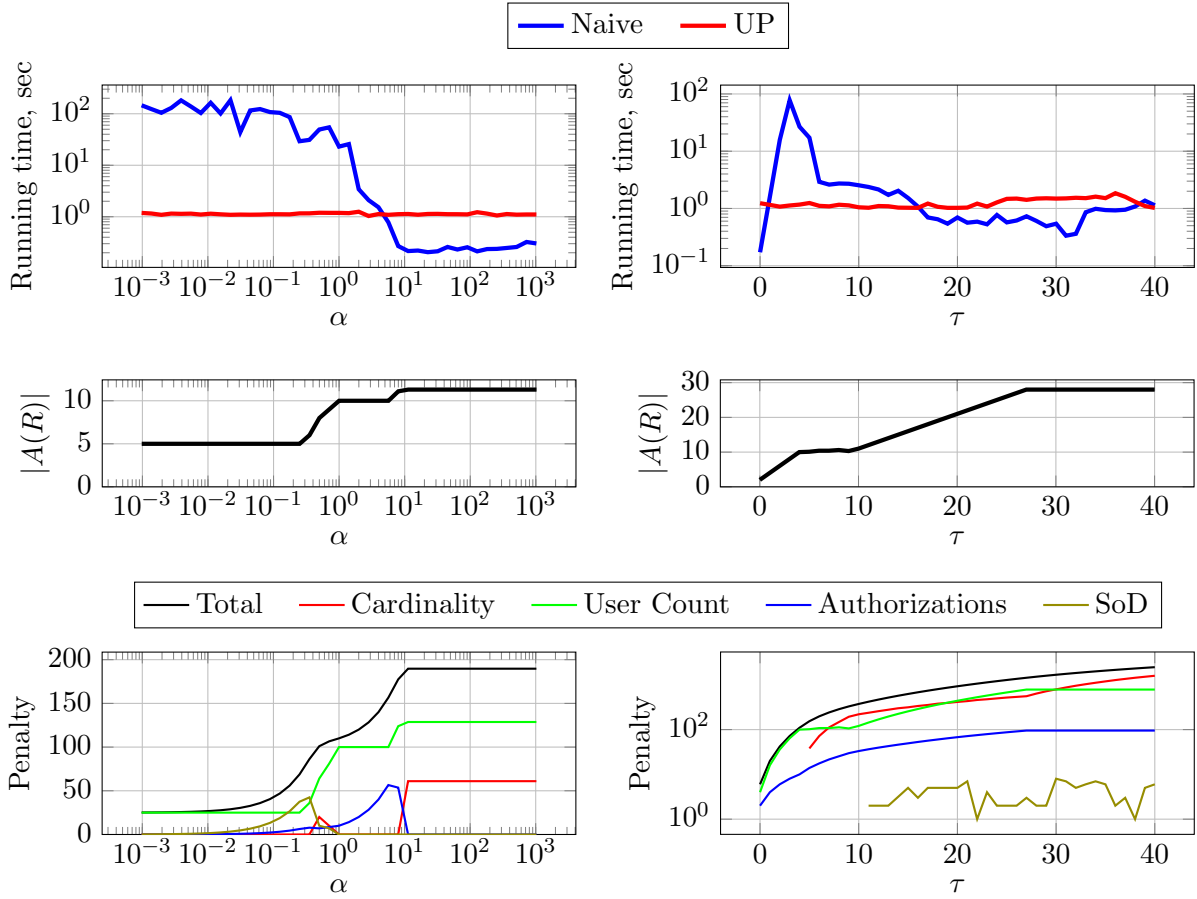
\begin{figure}[t]
\centering
	\begin{subfigure}[b]{1\linewidth}
		\begin{center}
		\begin{tikzpicture}
			\begin{customlegend}[
				legend columns=-1,
				legend style={align=left,/tikz/every even column/.append style={column sep=1em}},
        		legend entries={Naive, UP}]
	        \addlegendimage{naive solver, line legend}
	        \addlegendimage{up solver, line legend}
        	\end{customlegend}
		\end{tikzpicture}
		\end{center}
	\end{subfigure}


	\begin{subfigure}[t]{0.49\linewidth}
		\begin{center}
		\begin{tikzpicture}[trim axis right,trim axis left]
			\begin{loglogaxis}[
				height=4cm,
				xlabel={$\alpha$},
				ylabel={Running time, sec},
				ymin=0,
				my subfig
			]
	
	        \addplot[naive solver] table[
				x=alpha,
				y=time,
			] {times-alpha-naive.txt};
	
	        \addplot[up solver] table[
				x=alpha,
				y=time,
			] {times-alpha-up.txt};
			\end{loglogaxis}
		\end{tikzpicture}		
		\end{center}
	\end{subfigure}
	\hfill
	\begin{subfigure}[t]{0.49\linewidth}
		\begin{center}
		\begin{tikzpicture}[trim axis right,trim axis left]
			\begin{semilogyaxis}[
				height=4cm,
				xlabel={$\tau$},
				ylabel={Running time, sec},
				ymin=0,
				my subfig
			]
	
	        \addplot[naive solver] table[
				x=tau,
				y=time,
			] {times-tau-naive.txt};
	
	        \addplot[up solver] table[
				x=tau,
				y=time,
			] {times-tau-up.txt};
			\end{semilogyaxis}
		\end{tikzpicture}	
		\end{center}
	\end{subfigure}
	
	\bigskip
	
	\begin{subfigure}[t]{0.49\linewidth}
		\begin{center}
		\begin{tikzpicture}[trim axis right,trim axis left]
			\begin{semilogxaxis}[
				height=3cm,
				xlabel={$\alpha$},
				ylabel={$|A(R)|$},
				ymin=0,
				my subfig
			]
		
	        \addplot[black, ultra thick] table[
				x=alpha,
				y=users,
			] {times-alpha-up.txt};
			\end{semilogxaxis}
		\end{tikzpicture}		
		\end{center}
	\end{subfigure}
	\hfill
	\begin{subfigure}[t]{0.49\linewidth}
		\begin{center}
		\begin{tikzpicture}[trim axis right,trim axis left]
			\begin{axis}[
				height=3cm,
				xlabel={$\tau$},
				ylabel={$|A(R)|$},
				ymin=0,
				my subfig
			]
	
	        \addplot[black, ultra thick] table[
				x=tau,
				y=users,
			] {times-tau-up.txt};
			\end{axis}
		\end{tikzpicture}		
		\end{center}
	\end{subfigure}
	
	\bigskip
		
	\begin{subfigure}[b]{1\linewidth}
		\begin{center}
		\begin{tikzpicture}
			\begin{customlegend}[
				legend columns=-1,
				legend style={align=left,/tikz/every even column/.append style={column sep=1em}},
        		legend entries={Total, Cardinality, User Count, Authorizations, SoD}]
	        \addlegendimage{total penalty, line legend}
	        \addlegendimage{card penalty, line legend}
	        \addlegendimage{user count penalty, line legend}
	        \addlegendimage{auth penalty, line legend}
	        \addlegendimage{sod penalty, line legend}
        	\end{customlegend}
		\end{tikzpicture}
		\end{center}
	\end{subfigure}
	
	\smallskip
	
	\begin{subfigure}[b]{0.49\linewidth}
		\begin{center}
		\begin{tikzpicture}[trim axis right,trim axis left]
			\begin{semilogxaxis}[
				height=4cm,
				xlabel={$\alpha$},
				ylabel={Penalty},
				ymin=0,
				my subfig
			]
			\addplot[total penalty] table[
				x=alpha,
				y=objective,
			] {times-alpha-up.txt};	
	
	        \addplot[card penalty] table[
				x=alpha,
				y=card-penalty,
			] {times-alpha-up.txt};
	
	        \addplot[user count penalty] table[
				x=alpha,
				y=users-penalty,
			] {times-alpha-up.txt};
	
	        \addplot[auth penalty] table[
				x=alpha,
				y=authorisations-penalty,
			] {times-alpha-up.txt};
			
	        \addplot[sod penalty] table[
				x=alpha,
				y=sod-penalty,
			] {times-alpha-up.txt};
			\end{semilogxaxis}
			\end{tikzpicture}
		\end{center}

		\caption{Changing input $\alpha$ while $\tau = 4$.}
		\label{fig:alpha}
	\end{subfigure}
	\hfill
	\begin{subfigure}[b]{0.49\linewidth}
		\begin{center}
		\begin{tikzpicture}[trim axis right,trim axis left]
			\begin{semilogyaxis}[
				height=4cm,
				xlabel={$\tau$},
				ylabel={Penalty},
				ymin=0,
				my subfig
			]
			\addplot[total penalty] table[
				x=tau,
				y=objective,
			] {times-tau-up.txt};
	
	        \addplot[card penalty] table[
				x=tau,
				y=card-penalty,
			] {times-tau-up.txt};
	
	        \addplot[user count penalty] table[
				x=tau,
				y=users-penalty,
			] {times-tau-up.txt};
	
	        \addplot[auth penalty] table[
				x=tau,
				y=authorisations-penalty,
			] {times-tau-up.txt};
			
	        \addplot[sod penalty] table[
				x=tau,
				y=sod-penalty,
			] {times-tau-up.txt};
			\end{semilogyaxis}
		\end{tikzpicture}
		\end{center}

		\caption{Changing input $\tau$ while $\alpha = 1$.}
		\label{fig:tau}
	\end{subfigure}
	\caption{Analysis of the solution time, number $|A(R)|$ of users involved in the solution and the penalties as the instance generator inputs change.
	In all instances, $k = 8$ and $n = 80$.}
\label{fig:alpha-and-tau}
\end{figure}

The second set of experiments is designed to analyse the impact of the instance generator inputs $\alpha$ and $\tau$ on the instances, optimal solutions and the running times of the solvers.
The results are presented in Figure~\ref{fig:alpha-and-tau}.

These experiments reveal that the values of $\alpha$ and $\tau$ have little effect on the running time of UP\@.
In fact, the running time is consistently proportional to the size of the formulation $O(n \cdot 2^k)$.
Also, the composition of the formulation takes about half of the running time.
In other words, CPLEX solves this formulation in time linear in its size but the size of the formulation is exponential in $k$ putting a limit on how far this method can be scaled.

Thus, the Naive solver outperforms the UP solver in some extreme cases; when the instances are easy, the Naive formulation can exploit their special structure whereas the UP formulation remains large and as a result slow.
For example, when $\alpha$ is large, breaking the SoD constraints and authorizations becomes prohibitively expensive which significantly reduces the search space for the Naive solver.

However, when $\alpha$ is close to 1 and $\tau$ is small, the instances are particularly challenging as the optimal solutions tend to balance the penalties of all types, and this is where UP is particularly effective compared to Naive.

\subsubsection{Experiment conclusions}
For an organization that implements a workflow management system and has strict business continuity requirements, it will be important to find a trade-off between satisfying authorization policies and constraints and ensuring that workflow instances can complete when users are unavailable. We believe these experiments provide some useful insights into the interplay between authorization policies, separation of duty constraints and resiliency, and {form} a basis from which costs of violating policies and resiliency can be balanced.

The trade-offs between authorization and resiliency requirements are evident in Figure~\ref{fig:alpha}. For very small values of $\alpha$ (when penalties for violating authorization requirements are low) the penalties in Valued APEP solutions are dominated by $|A(R)|$, the number of users assigned to the extended plan. As $\alpha$ increases, the penalties associated with violations of authorization and constraints begin to dominate, and, as $\alpha$ increases further (meaning that authorizations and constraints become increasingly expensive to violate) the penalties associated with breaking resiliency and $|A(R)|$ dominate. We also see that $|A(R)$ reaches a maximum value when $\alpha$ equals around $10$, at which point the penalties associated with violating authorization requirements are negligible compared to those associated with resiliency and $|A(R)|$.

From Figure~\ref{fig:tau} we see that for {$\tau \le 4$} {resiliency requirements are always satisfied}. For {$\tau > 4$}, unsurprisingly, $\tau$-resilient extended plans would have to be a large proportion of the user population and the penalties associated with such plans begin to dominate. {For very large values of $\tau$} we find that $|A(R)|$ drops below the value of $\tau$, meaning the extended plans cannot be $\tau$-resilient and the penalties for such plans are mainly associated with resiliency and the  of the solution.

Finally, we note that the concept of user profiles, used to prove theoretical results about \vapep{}, also enables us to derive a MIP formulation that can be solved by CPLEX in FPT-like time and that is efficient across all our test instances.
In contrast, the straightforward (Naive) formulation of the problem scales super-polynomially with the size of the problem even if the small parameter is fixed.
This shows the importance of the FPT studies even if the researcher intends to use general-purpose solvers to address the problem.
%
%
%

\section{Related work and discussion}\label{sec:related-work}

{\sc Valued APEP} builds on a number of different strands of recent research in access control, including workflow satisfiability, workflow resiliency and risk-aware access control.
Workflow satisfiability is concerned with finding an allocation of users to workflow steps such that every user is authorized for the steps to which they are assigned and all workflow constraints are satisfied.
Work in this area began with the seminal paper by Bertino {\em et al.}~\cite{BeFeAt99}.
Wang and Li initiated the use of parameterized complexity analysis to better understand workflow satisfiability~\cite{WaLi10}, subsequently extended to include user-independent constraints~\cite{CoCrGaGuJo14} and the study of {\sc Valued WSP}~\cite{CramptonGK15}.
As we have seen APEP can be used to encode workflow satisfiability problems.

Workflow resiliency is concerned with ensuring business continuity in the event that some (authorized) users are unavailable to perform steps in a workflow~\cite{LiWaTr09,MaMoMo14,Fong19}.
Berg\'e {\em et al.} showed that APEP can be used to encode certain kinds of resiliency policies~\cite[Section 6]{BergeCGW18}.

Researchers in access control have recognized that it may be necessary to violate access control policies in certain, exceptional circumstances~\cite{MaDuSl14,Petritsch14}, provided that those violations are controlled appropriately.
One means of controlling violations is by assigning a cost to policy violations, usually defined in terms of risk~\cite{ChCr11,DiBeEyBaMo04}.
Thus, the formalization of problems such as {\sc Valued WSP} and {\sc Valued APEP} and the development of algorithms to solve these problems may be of use in developing risk-aware access control systems.

Thus, we believe that APEP and {\sc Valued APEP} are interesting and relevant problems, and understanding the complexity of these problems and developing the most efficient algorithms possible to solve them is important.
A considerable amount of work has been done on the complexity of WSP, showing that the problem is FPT for many important classes of constraints~\cite{CrGuYe13,CoCrGaGuJo14,KarapetyanPGG19}.
It is also known that {\sc Valued WSP} is FPT and, for user-independent constraints, the complexity of the problem is identical to that for WSP (when polynomial terms in the sizes of the user set and constraint set are disregarded in the running time)~\cite{CramptonGK15}.
Roughly speaking, this is because (weighted) user-independent constraints in the context of workflow satisfiability allow us to restrict our attention to partitions of the set of steps when searching for solutions, giving rise to the exponential term $2^{k \log k}$ in the running time of an algorithm to solve ({\sc Valued}) WSP.

APEP, unsurprisingly, is known to be a more complex problem~\cite{BergeCGW18}.
The complexity of APEP differs from WSP because it is not sufficient to consider partitions of the set of resources, in part because an arbitrary relation $A$ is not a function.
The results in this paper provide the first complexity results for {\sc Valued APEP}, showing (in Corollary~\ref{cor:FPT_known_constraints}) that it is no more difficult than APEP for constraints in $\bodU$, $\bodE$, $\sodU$ and $\sodE$ (disregarding polynomial terms).

We believe the concept of a user profile and Theorem~\ref{thm:t_bounded_algorithm} are important contributions to the study of APEP as well as {\sc Valued APEP}, providing a generic way of establishing complexity results for different classes of constraints.
In particular, Corollary~\ref{cor:FPT_known_constraints} of Theorem~\ref{thm:t_bounded_algorithm} actually shows how to improve existing results for APEP$\langle \bodU,\bodE,\sodE,\sodU \rangle$ due to Berg\'e {\em et al}~\cite{BergeCGW18}.
Moreover, when an APEP instance is equivalent to a WSP instance (i.e, it contains a cardinality constraint $(r,\leq, 1)$ for each $r \in R$) then the instance is $k$-bounded, and a user profile is the characteristic function of some partition of $R$.
Thus we essentially recover the known FPT result for {\sc Valued WSP}, which is based on the enumeration of partitions of the set of workflow steps.

\section{Concluding remarks}
\label{sec:conclusions}

We believe this paper makes three significant contributions.
First, we introduce {\sc Valued APEP}, a generalization of APEP, which, unlike APEP, always returns some authorization relation.
Thus a solution to {\sc Valued APEP} is more useful than that provided by APEP:
if there exists a valid authorization relation {\sc Valued APEP} will return it; if not, {\sc Valued APEP} returns a solution of minimum weight.
This allows an administrator, for example, to decide whether to implement the solution for an instance of {\sc Valued APEP} or adjust the base authorization relation and/or the constraints in the input in an attempt to find a more appropriate solution.

The second contribution is to advance the techniques available for solving APEP as well as {\sc Valued APEP}.
Specifically, the notion of a user profile plays a similar role in the development of algorithms to solve ({\sc Valued}) APEP as patterns do in solving ({\sc Valued}) WSP.
The enumeration of user profiles is a powerful technique for analyzing the complexity of {\sc Valued} APEP, yielding general results for the complexity of the problem (which are optimal assuming the Exponential Time Hypothesis holds) and improved results for APEP.

The third contribution is the experimental study that involves a new set of realistic benchmark instances, two mixed integer programming formulations of \vapep{} and extensive analysis of the computational results.
Apart from the conclusions related to the new concept of $\tau$-resiliency in workflows, we demonstrate that a general-purpose solver can solve an FPT problem in FPT-like time if the formulation is `FPT-aware', i.e., if it exploits our understanding of the FPT properties of the problem, and that such an `FPT-aware' formulation significantly outperforms a naive formulation.
This is particularly significant for practitioners who often prefer to use general purpose solvers to address complex problems, as they can benefit from the otherwise theoretical computational complexity studies.

There are several opportunities for further work.
We  intend to investigate other (weighted) user-independent constraints for ({\sc Valued}) APEP.
First, we are interested in what other problems in access control can be encoded as APEP instances, apart from workflow satisfiability and resiliency problems.
Second, we would like to consider appropriate weight functions for such encodings, which would have the effect of providing more useful (weighted) solutions for the original problems (rather a binary yes/no solution).
Our work also paves the way for work on quantifying the trade-offs associated with violating security and resiliency requirements when it is impossible to satisfy both simultaneously. 
A better understanding of these trade-offs together with tools for computing optimal solutions would seem to have considerable value to commercial organizations, enabling them to manage conflicting security and business requirements in an informed manner.

\paragraph{Acknowledgement} Research in this paper was supported by Leverhulme Trust grant RPG-2018-161.



\end{document}